\documentclass[conference]{IEEEtran}
\IEEEoverridecommandlockouts
\usepackage{float}
\usepackage{graphicx}
\usepackage{clrscode}
\usepackage{stfloats}
\usepackage{paralist}
\usepackage{tikz}
\usepackage{soul} 
\usepackage{color}
\usepackage{cite}
\usepackage{multicol}
\usepackage{comment}
\usepackage{amsmath, amssymb, mathrsfs, amsfonts}
\usepackage{amsthm}
\usepackage{mathtools}
\usepackage{epsfig, epstopdf}
\usepackage{algorithm}
\usepackage{enumitem}
\usepackage{algorithmic}
\usepackage{lipsum}
\usepackage{pbox}
\usepackage{array}
\usepackage{tablefootnote}
\usepackage{subfigure}

\usetikzlibrary{arrows,automata}

\allowdisplaybreaks

\newtheorem{lemma}{Lemma}

\usepackage{hyperref}

\newcommand{\RNum}[1]{\uppercase\expandafter{\romannumeral #1\relax}}
\def\BibTeX{{\rm B\kern-.05em{\sc i\kern-.025em b}\kern-.08em
    T\kern-.1667em\lower.7ex\hbox{E}\kern-.125emX}}

\begin{document}

\title{On Stochastic Performance Analysis of Secure Integrated
Sensing and Communication Networks}
\author{\IEEEauthorblockN{Marziyeh Soltani, Mahtab Mirmohseni, \textit{Senior Member, IEEE}, and Rahim Tafazolli, \textit{Fellow, IEEE}}  
\IEEEauthorblockA{\textit{5/6GIC, Institute for Communication Systems, University of Surrey}, Guildford, U.K \\  
Emails: \{m.soltani, m.mirmohseni, r.tafazolli\}@surrey.ac.uk}  
}

\maketitle
\begin{abstract}
This paper analyzes the stochastic security performance of a multiple-input multiple-output (MIMO) integrated sensing and communication (ISAC) system in a downlink scenario. A base station (BS) transmits a multi-functional signal to simultaneously communicate with a user, sense a target’s angular location, and counteract eavesdropping threats. The system includes a passive single-antenna communication eavesdropper and a multi-antenna sensing eavesdropper attempting to infer the target’s location. The BS-user and BS-eavesdroppers channels follow Rayleigh fading, while the target's azimuth angle is uniformly distributed. To evaluate the performance, we derive exact expressions for the secrecy ergodic rate and the ergodic Cramér-Rao lower bound (CRB) for target localization at both the BS and the sensing eavesdropper. This involves computing the probability density functions (PDFs) of the signal-to-noise ratio (SNR) and CRB, leveraging the central limit theorem for tractability. Numerical results validate our findings.
\end{abstract}
\section{Introduction}\label{introduction}
Recently, integrated sensing and communications (ISAC) has emerged as a key paradigm for future wireless networks, enabling applications such as vehicular communication and smart city infrastructure \cite{Integratedtoward}. By sharing resources like spectrum and hardware, ISAC enhances both sensing and communication (S \& C). As 6G ISAC systems handle increasingly complex S \& C data exchanges, security and privacy concerns have intensified. Traditional security methods, including encryption, struggle with ISAC’s latency-sensitive nature and potential leakage of transmission parameters \cite{IntegratingSensingandCommunicationsin6G}. This has driven interest in physical layer security (PLS), which provides real-time protection without complex key management.  

Compared to traditional systems that only handle communication, ISAC systems face more complex security issues. These challenges mainly fall into two categories: \textit{keeping communication data secure} and \textit{sensing security in radar-like sensing} \cite{SecuringtheSensingFunctionalityinISACNetworks}. The first challenge, comes from the fact that ISAC systems use waveforms that carry information. These waveforms can be picked up by the targets being sensed, which might act as eavesdroppers and access sensitive information. To mitigate this risk, \cite{SecureIntegratedSensingandCommunication} uses information-theoretic approaches, investigating both the inner and outer bounds of the secrecy-distortion region. Also, in \cite{JointSecureTransmitBeamformingDesignsforIntegratedSensingandCommunicationSystems,PhysicalLayerSecurityOptimizationWithCramér–RaoBoundMetric,Sensing-AssistedEavesdropperEstimation:AnISACBreakthroughinPhysicalLayerSecurity,OptimalBeamformingforsecureIntegratedSensingandCommunicationExploitingTargetLocation}, various beamforming techniques improve PLS by minimizing eavesdroppers’ signal quality while preserving radar performance. Additionally, covert ISAC is studied in \cite{RobustBeamformingDesignforCovertIntegratedSensingandCommunication}. 

\textit{Sensing security} in ISAC is another concern, as adversaries can intercept reflected sensing signals to infer system activities \cite{OnRadarPrivacyinSharedSpectrumScenarios}. To mitigate this, approaches such as mutual information constraints \cite{SecuringtheSensingFunctionalityinISACNetworks} and beamforming optimization \cite{SecureCell-FreeIntegratedSensingandCommunicationinthePresenceofInformation} have been proposed to degrade adversarial sensing capabilities. Furthermore, \cite{IllegalSensingSuppressionforIntegratedSensingandCommunicationSystem} aims to protect target privacy by reducing the adversary’s detection accuracy.  

Existing ISAC security solutions often rely on nonconvex optimization, requiring iterative algorithms like successive convex approximation (SCA) and semidefinite relaxation (SDR), which lack closed-form beamforming solutions \cite{SecureDual-FunctionalRadar-CommunicationTransmission:ExploitingInterference, JointSecureTransmitBeamformingDesignsforIntegratedSensingandCommunicationSystems}. Motivated by this, we propose a precoding matrix with a closed-form expression. The intuition behind this matrix is that it reduces to the optimal beamforming solution in the absence of security concerns (i.e., no eavesdropper); and also to the optimal beamforming solution in the absence of sensing requirements (i.e., no target). Then, we use this beamforming and focus on deriving the performance metrics.  

While random ISAC networks without security constraints have been studied \cite{OnStochasticFundamentalLimitsinaDownlink}, to the best of our knowledge, this is the first paper to analyze ISAC communication and sensing security while considering channel randomness. We consider a downlink multiple-input multiple-output (MIMO) ISAC system where a multi-antenna base station (BS) transmits a precoding matrix with artificial noise to:  
1. Deliver communication data to a single-antenna user,  
2. Sense a target’s angular location via echo signals, and  
3. Prevent a passive, single-antenna eavesdropper from intercepting user communications through wiretap coding. Additionally, we account for target privacy by considering a multi-antenna sensing eavesdropper extracting the target’s angular information. The BS-user and BS-eavesdropper channels follow independent Rayleigh fading models, while the target's azimuth angle is uniformly distributed.\footnote{Our approach applies to any target angle distribution.} Our main contribution is the derivation of exact expressions for the ergodic secrecy rate, as well as a closed-form expression for the ergodic Cramér-Rao Bound (CRB), defined as \( E[\text{CRB}] \), for the target’s angle of arrival at the BS and the sensing eavesdropper\footnote{The CRB provides a lower bound on the mean squared error (MSE) of unbiased estimators.}. We have previously showen in \cite{OnStochasticFundamentalLimitsinaDownlink} that the CRB is a tighter bound than the Bayesian CRB (BCRB)\footnote{These metrics were highlighted as a future research direction in \cite{recentadvancesandtenopenchallenges}.}, making it a crucial metric for ISAC security and privacy analysis in random networks.

\textbf{Notation:} Scalars, vectors, and matrices are denoted by lowercase, boldface lowercase, and boldface uppercase letters, respectively. \( P(\cdot) \), \( f_x(\cdot) \), and \( E[\cdot] \) represent probability, PDF, and expectation. \( \mathbf{X}^{M \times N} \), \( \mathbf{X}^T \), \( \mathbf{X}^H \), and \( \mathbf{X}^* \) denote an \( M \times N \) matrix, its transpose, Hermitian transpose, and conjugate. The Euclidean norm is \( \|\cdot\| \), and the complex norm is \( |\cdot| \). \( \mathcal{CN}(\cdot,\cdot) \) and \( \mathcal{U}(a,b) \) denote circularly symmetric complex Gaussian and uniform distributions, respectively. \( \mathcal{N}_3(\boldsymbol{\mu}, \boldsymbol{\Sigma}) \) is a trivariate normal distribution with mean \( \boldsymbol{\mu} \) and covariance \( \boldsymbol{\Sigma} \). \( \mathbb{C} \) and \( \mathbb{R} \) denote complex and real number sets. \( \overset{d}{\rightarrow} \) and \( \overset{p}{\rightarrow} \) indicate convergence in distribution and probability. \( \text{Tr}(A) \), \( \otimes \), \( \text{vec} \), and \( \mathbf{I}_m \) represent the trace of \( A \), Kronecker product, vectorization, and the \( m \times m \) identity matrix.
\section{System Model and Transmission Strategy}\label{systemmodel}
We consider a BS with \(N\) transmit and \(M\) receive antennas serving a single-antenna communication user in the downlink. The BS knows the user’s channel and simultaneously senses a distant point target with an unknown location with a known distribution. Using a monostatic radar setup, the BS colocates its estimator and transmit antennas, ensuring identical direction of arrival (DoA) and departure (DoD). Two eavesdroppers are present:  
1. A single-antenna communication eavesdropper passively overhearing transmissions, with an unknown channel to the BS.
2. A sensing eavesdropper with (\(N_e\)) receive antennas, estimating the target’s position, threatening the privacy, and assumed to be strong with knowledge of the BS’s transmitted data. The system model is shown in Fig. \ref{systemmodelfig}\footnote{The case of weak eavesdropper is left for future work.}.
\begin{figure}
    \includegraphics[scale=.63]{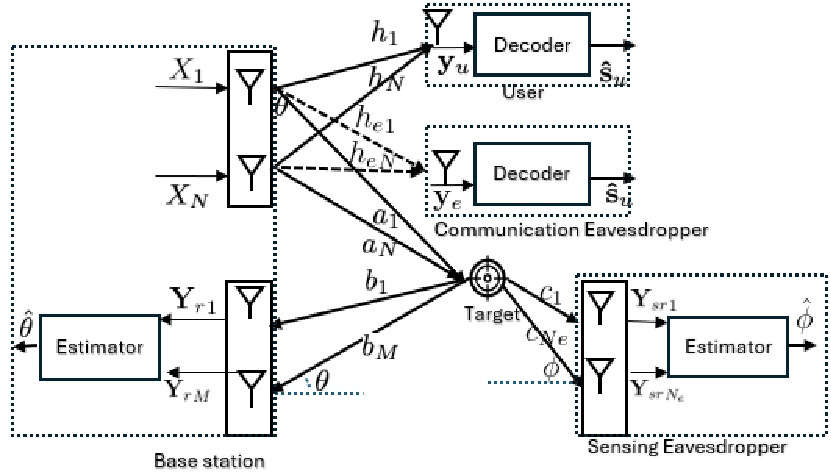}\caption{System model.${\mathbf{Y}_{sr}}_{i}$ and $h_{ei}$ denotes the $i$-th row and element of $\mathbf{Y}_{sr}$ and $\mathbf{h}_{e}$, respectively} \label{systemmodelfig}
\end{figure}

\textbf{Channels:} The channel vector between the BS and the user is given by \(\mathbf{h}=[h_1\quad h_2 \dots h_N]^T \in \mathbb{C}^{N \times 1}\), where the elements are independent and identically distributed (i.i.d.) according to \(\mathcal{CN} (0,1)\). More precisely, the \(i\)-th element can be expressed as \({h}_i = |{h}_i|e^{j{{\phi_i}}}\), where \(|{h}_i|\) follows a Rayleigh distribution with a scale parameter of \(1\), and \({{\phi_i}}\) is uniformly distributed over \([0, 2\pi)\). Furthermore, assuming an even number of antennas, the transmit and receive array steering vectors from the BS to the target are shown as follows:  
\begin{align}
\mathbf{a}^{N \times 1}(\theta) & \! \!= \! \![e^{-j\pi \sin(\theta)\frac{N-1}{2}}, e^{-j\pi \sin(\theta)\frac{N-3}{2}}, \dots, e^{j\pi \sin(\theta)\frac{N-1}{2}} ]^T,\nonumber\\
\mathbf{b}^{M \times 1}(\theta) & \! \!=  \! \![ e^{-j\pi \sin(\theta)\frac{M-1}{2}}, e^{-j\pi \sin(\theta)\frac{M-3}{2}}, \dots, e^{j\pi \sin(\theta)\frac{M-1}{2}}]^T,\nonumber
\end{align}
where \(\theta\) represents the azimuth angle of the target relative to the BS, which follows $\mathcal{U}(-\pi/2,\pi/2)$. The \(i\)-th element of \(\mathbf{a}(\theta)\) is expressed as \(a_i=e^{-jf_i}\), where \( f_i=\pi \sin(\theta)\frac{N-(2i-1)}{2} \). Moreover, the channel vector between the BS and the communication eavesdropper is given by \(\mathbf{h}_e \sim\mathcal{CN}(0,\mathbf{I}) \in \mathbb{C}^{N \times 1}\). The received steering vector at the sensing eavesdropper\footnote{For simplicity, we ignore the link between the BS and the sensing eavesdropper.} is given by: $\mathbf{c}^{N_e \times 1}(\phi) =[ e^{-j\pi \sin(\phi)\frac{N_e-1}{2}}, e^{-j\pi \sin(\phi)\frac{N_e-3}{2}}, \dots, e^{j\pi \sin(\phi)\frac{N_e-1}{2}}]^T$, 
where \(\phi\) represents the azimuth angle of the target relative to the sensing eavesdropper, which follows $\mathcal{U}(-\pi/2,\pi/2)$.
\textbf{ISAC Transmission Strategy:} Optimizing beamforming for secure ISAC lacks a closed-form solution, as shown in \cite{PhysicalLayerSecurityOptimizationWithCramér–RaoBoundMetric,Sensing-AssistedEavesdropperEstimation:AnISACBreakthroughinPhysicalLayerSecurity,OptimalBeamformingforsecureIntegratedSensingandCommunicationExploitingTargetLocation}. However, two special cases allow closed-form solutions: 1) No eavesdroppers (no security concerns): The optimal beamforming vector minimizing the CRB for target estimation—while ensuring a maximize communication rate under a power constraint—lies in the span of \(\{\mathbf{a}, \mathbf{h}\}\), as shown in \cite[Lemma 1]{CrameRaoBoundOptimizationforJoint}.  
2) No target (no sensing task) \& unknown eavesdropper’s CSI: When the eavesdropper’s CSI is unavailable, secrecy performance metrics include the ergodic secrecy rate and secrecy outage-based secrecy rate \cite{Artificial-Noise-AidedBeamformingDesignintheMISOME}. Regarding secrecy ergodic rate, a near-optimal strategy is the masked beamforming method \cite{SecureTransmissionWithMultipleAntennasI}, where the beamforming vector is based only on the legitimate user’s channel. In this method, the transmitter Sends the message (encoded with a scalar Gaussian wiretap code) along the user’s channel to maximize the signal power. Also, it injects spatio-temporal white noise into the null space of the user’s channel, avoiding interference while degrading the eavesdropper’s reception. Regarding secrecy outage, artificial noise (AN) distributed uniformly in the null space of the information beam is proven optimal \cite{SecrecyOutageinMISOSystemsWithPartialChannelinformation}.

Motivated by this discussion, we propose adopting an AN-based transmit beamforming structure that, while suboptimal, has a closed-form expression. Proposing a structured transmit beamformer, which is a weighted combination of two different precoders, is a novel technique for analyzing the performance of ISAC, especially when non-convex problems result in a non-closed-form and computationally challenging solution for the optimal beamformer, as seen in \cite{PowerAllocationforMassiveMIMO-ISACSystems} for massive MIMO-ISAC. The BS transmits:
\begin{align}
\mathbf{X}=\sqrt{P\tau}\mathbf{t}_1 \mathbf{s}_u+ \sqrt{P(1-\tau)}\mathbf{G} \mathbf{V} \in \mathbb{C}^{N\times L},\label{x}
\end{align}
where the first term is optimal without considering security \cite[Lemma 1]{CrameRaoBoundOptimizationforJoint}, while the second term, combined with part of the first term’s power, is optimal when sensing is disregarded\cite{SecureTransmissionWithMultipleAntennasI,SecrecyOutageinMISOSystemsWithPartialChannelinformation}. \(L > N\) represents the length of the radar pulse or communication frame; \(\mathbf{s}_u \in \mathbb{C}^{1 \times L}\) is a white Gaussian signaling data stream for the user with unit power, satisfying \(\frac{1}{L}E\{\mathbf{s}_u\mathbf{s}_u^H\} \approx 1\) when \(L\) is sufficiently large \cite{CrameRaoBoundOptimizationforJoint}. The matrix \(\mathbf{V} \in \mathbb{C}^{(N-2) \times L}\) represents an artificial noise (AN), whose elements are i.i.d. complex Gaussian, such that each row, denoted by \(\mathbf{v}_i\), satisfies \(\frac{1}{L}\mathbf{v}_i\mathbf{v}_i^H={\frac{1}{N-2}}\). We assume that the AN signals and data \(\mathbf{s}_u\) are orthogonal, i.e., \(\mathbf{s}_u\mathbf{v}_i^H=0\). The vectors \(\mathbf{t}_1\) and \(\mathbf{G}\) are constructed as follows. First, we construct an orthonormal basis for the \(N\)-dimensional space. The first basis vector is aligned with the target channel: \(\tilde{\mathbf{a}}=\frac{\mathbf{a}}{||\mathbf{a}||}\). The second basis vector is \(\tilde{\mathbf{h}}=\frac{\mathbf{h}-(\tilde{\mathbf{a}}^H \mathbf{h})\tilde{\mathbf{a}}}{||\mathbf{h}-(\tilde{\mathbf{a}}^H \mathbf{h})\tilde{\mathbf{a}}||}\). We choose \(\mathbf{t}_1\) to be a vector in the span of \(\tilde{\mathbf{a}}\) and \(\tilde{\mathbf{h}}\), i.e., \(\mathbf{t}_1= \alpha \tilde{\mathbf{a}} + \beta \tilde{\mathbf{h}}\). Next, defining the matrix \(\mathbf{A}= \begin{bmatrix} \mathbf{\tilde{a}} & \mathbf{\tilde{h}} \end{bmatrix} \in \mathbb{C}^{N \times 2}\), an orthonormal basis for the null space of \( \mathbf{A}^H \), i.e., all vectors \( \mathbf{x} \) such that \( \mathbf{A}^H \mathbf{x} = 0 \), is derived as follows. We compute the singular value decomposition (SVD) \(\mathbf{A}= \mathbf{U} \mathbf{\Sigma} \mathbf{\tilde{W}}^H\). Then, the last \(N - 2\) columns of \( \mathbf{U} \), denoted by \(\mathbf{t}_3,...,\mathbf{t}_N\), form the basis for the null space of \(\mathbf{A}^H\). The matrix \(\begin{bmatrix} \mathbf{\tilde{a}} & \mathbf{\tilde{h}} & \mathbf{G} \end{bmatrix}\), where \(\mathbf{G}=\begin{bmatrix} \mathbf{t}_3 & ... & \mathbf{t}_N \end{bmatrix}\), constructs an orthonormal basis for the \(N\)-dimensional space. The sample covariance matrix, due to the orthogonality of \(\mathbf{v}_i\) and \(\mathbf{s}_u\), is:  
\begin{align}
\mathbf{R}_x=\frac{1}{L}\mathbf {X}\mathbf{X}^H\approx 
P\tau\mathbf{t}_1\mathbf{t}^H_1+\frac{(1-\tau)P}{N-2}\sum_{i=3}^{N}\mathbf{t}_i\mathbf{t}^H_i.\label{Rx}
\end{align}  
We assume that the BS power is \( p_t \); Thus: $p_t= \text{Tr}(R_x) = P\tau\text{Tr}(\mathbf{t}_1\mathbf{t}_1^H) + \frac{(1-\tau)P}{N-2} \text{Tr} \left( \sum_{i=3}^{N} \mathbf{t}_i\mathbf{t}_i^H \right) \overset{(a)}{=} P.$ Here, (a) follows from the unitary property of \(\mathbf{U}\) and the orthonormality of \(\mathbf{\tilde{a}}\) and \(\mathbf{\tilde{h}}\), ensuring that \( P \) represents the BS transmit power, and \( \tau \) denotes the fraction allocated to $\mathbf{s}_u$.  

\textbf{ Received signals:} The received signals at the user and the communication eavesdropper are given by \( \mathbf{y}_u = c_1\mathbf{h}^H \mathbf{X} + \mathbf{z}_u \) and \( \mathbf{y}_{e} = c_2\mathbf{h}_e^H \mathbf{X} + \mathbf{z}_e \), respectively, where \( \mathbf{z}_u \) and \( \mathbf{z}_e \in \mathbb{C}^{1 \times L} \) are additive white Gaussian noise (AWGN) vectors, with each element following the distribution \( \mathcal{CN} (0,\sigma^2) \)\footnote{For simplicity, we assume the noise terms have the same variance.}. Here, \( c_1 \) and \( c_2 \) denote the complex-valued path gains (functions of distance) from the BS to the user and the eavesdropper, respectively. When the BS transmits \( \mathbf{X} \) to sense the target, the BS and the sensing eavesdropper receive the reflected echo signal matrices as $\mathbf{Y}_r = c_3 \mathbf{b}(\theta)\mathbf{a}(\theta)^H \mathbf{X} + \mathbf{Z}_r$ and $\mathbf{Y}_{sr} = c_4 \mathbf{c}(\phi)\mathbf{a}(\theta)^H \mathbf{X} + \mathbf{Z}_{sr},$
respectively, where \( c_3, c_4 \in \mathbb{C} \) are the complex-valued channel coefficients, which depend on the target’s radar cross-section (RCS) and the round-trip path loss between the target-BS and target-sensing eavesdropper, respectively. \( \mathbf{Z}_r \in \mathbb{C}^{M \times L} \) and \( \mathbf{Z}_{sr} \in \mathbb{C}^{N_e \times L} \) are AWGN matrices, with elements that are i.i.d. and follow the distribution \( \mathcal{CN} (0,\sigma^2_r) \).

\textbf{Metrics:} The metric used to evaluate sensing performance at the BS is \( \text{CRB}(\theta) \). To evaluate privacy, we use \( \text{CRB}(\phi) \), which indicates how accurately the sensing eavesdropper can estimate the target's angle relative to itself. Secure communication performance is assessed using the secrecy rate, which requires deriving the rate for the user, \( R \), and the leakage to the eavesdropper, \( R_e \). Given the random nature of the channels, \( \text{CRB}(\theta) \), \( \text{CRB}(\phi) \), \( R \), and \( R_e \) are all random variables, underscoring the need for suitable performance metrics for analyzing random networks. To address this, we use the ergodic CRB, defined as \( E[\text{CRB}(\theta)] \) at the BS, and \( E[\text{CRB}(\phi)] \) at the sensing eavesdropper. In \cite{OnStochasticFundamentalLimitsinaDownlink}, we have shown that ergodic CRB serves as a lower bound for the estimation error and provides a tighter bound compared to Bayesian or deterministic CRBs. Moreover, since the channel of the sensing eavesdropper is unknown to the BS, the ergodic secrecy rate is \cite{SecureTransmissionWithArtificialNoiseOverFadingChannels}: $C_s = (E[R]-E[R_e])^+
=(E_{\mathbf{h},\mathbf{a}(\theta)} \left[ \log(1 + \text{SINR}_u) \right] - E_{\mathbf{h},\mathbf{a}(\theta),\mathbf{h_e}} \left[ \log(1 + \text{SINR}_e) \right])^+.$
\section{ CRB analysis}\label{crbanalysis}
We first derive the complementary cumulative distribution function (CCDF) of the CRBs, namely \( P(\text{CRB}(\theta) > \epsilon) \) and \( P(\text{CRB}(\phi) > \epsilon) \), which serve as the foundation for deriving \( E[\text{CRB}(\theta)] \) and \( E[\text{CRB}(\phi)] \), respectively. 
\subsection{CCDF of $\text{CRB}(\theta)$ at the BS}\label{opoftargetlb}
To derive the CRB of $\theta$, which is our sensing parameter of interest to be estimated at the BS, we use the received echo signal at the BS which is $\mathbf{Y}_{r} = c_3 \mathbf{b}(\theta)\mathbf{a}(\theta)^H \mathbf{X} + \mathbf{Z}_{r},$ where $\mathbf{X}=\sqrt{P\tau}\mathbf{t}_1 \mathbf{s}_u+ \sqrt{P(1-\tau)}\mathbf{G} \mathbf{V} = \sqrt{P\tau}\mathbf{t}_1 \mathbf{s}_u+ \sqrt{P(1-\tau)}\sum_{i=3}^{N}\mathbf{t}_i\mathbf{v}_i\in \mathbb{C}^{N\times L}$, and we obtain the Fisher information matrix (FIM) for estimating \( \xi=[\theta,\mathcal{R}(c_3),\mathcal{I}(c_3)]^T \in \mathbb{R}^{3 \times 1} \). Let \( \mathbf{A}(\theta) =\mathbf{b}(\theta) \mathbf{a}(\theta)^T \), the received echo signal at the BS can be rewritten as $\mathbf{Y}_{r} = c_3 \mathbf{A}(\theta) \mathbf{X} + \mathbf{Z}_{r}$. By vectorizing $\mathbf{Y}_{r}$, we have: $\tilde{\mathbf{y}}_{r} = \text{vec}(\mathbf{Y}_{r}) = \tilde{\mathbf{u}} + \tilde{\mathbf{n}},$ where \( \tilde{\mathbf{u}} = c_3 \text{vec}(\mathbf{A}\mathbf{X}) \) and \( \tilde{\mathbf{Z}_{r}} = \text{vec}(\mathbf{Z}_{r}) \sim \mathcal{CN}(0, \sigma^2_r\mathbf{I}_{ML}) \). Thus, we have: $\tilde{\mathbf{y}}_{r}\sim \mathcal{CN}(\tilde{\mathbf{u}}, \sigma^2_r\mathbf{I}_{ML})$. Let \( {\mathbf{F}} \in \mathbb{R}^{3 \times 3} \) denote the FIM for estimating $\xi$ based on $\tilde{\mathbf{y}}_{r}$. Each element of \( {\mathbf{F}} \) is given by \cite{Fundamentalsofstatisticalsignalprocessing}:
\begin{equation}
{\mathbf{F}}_{i,j} = \text{tr} \left\{ \mathbf{R}^{-1} \frac{\partial \mathbf{R}}{\partial \xi_i} \mathbf{R}^{-1} \frac{\partial \mathbf{R}}{\partial \xi_j} \right\}
+ 2 \Re \left\{ \frac{\partial \tilde{\mathbf{u}}^H}{\partial \xi_i} \mathbf{R}^{-1} \frac{\partial \tilde{\mathbf{u}}}{\partial \xi_j} \right\},\label{elemntoffisher}
\end{equation}
for $i,j \in \{1,2,3\}$ where $\mathbf{R}$ is the covariance matrix of the Gaussian observation which in our case is $\sigma^2_r\mathbf{I}_{ML}$. By defining $\bar{\alpha}=[\mathcal{R}(c_3),\mathcal{I}(c_3)]^T \in \mathbb{R}^{2 \times 1}$, the FIM \( {\mathbf{F}} \) is partitioned as: ${\mathbf{F}} = 
\begin{bmatrix}
{\mathbf{F}}_{\theta\theta} & {\mathbf{F}}_{\theta\bar{\alpha}} \\
{\mathbf{F}}_{\bar{\alpha}{\theta}} & {\mathbf{F}}_{\bar{\alpha}\bar{\alpha}}
\end{bmatrix}.$
The covariance matrix \( \mathbf{R} \) is independent of \( \xi \). Therefore, we have \( \frac{\partial \mathbf{R}}{\partial \xi_i} = 0 \), \( i = 1,2,3 \) and the first term in (\ref{elemntoffisher}) is zero. Furthermore, $\frac{\partial \tilde{\mathbf{u}}}{\partial \theta} = c_3 \text{vec}(\mathbf{\dot{A}} \mathbf{X})+c_3 \text{vec}(\mathbf{A} \mathbf{\dot{X}})$ and $\frac{\partial \tilde{\mathbf{u}}}{\partial \bar{\alpha}} = [1, j] \otimes \text{vec}(\mathbf{A} \mathbf{X})$. Accordingly, the elements of ${\mathbf{F}}$ are given in (\ref{elemntoffisher2}), where \( j = \sqrt{-1} \), and \( \dot{\mathbf{A}} = \frac{\partial \mathbf{A}}{\partial \theta} \) and \( \dot{\mathbf{X}} = \frac{\partial \mathbf{X}}{\partial \theta} \) denote the partial derivative of \( \mathbf{A} \) and $\mathbf{X}$ w.r.t. \( \theta \), respectively. In (\ref{elemntoffisher2}), (a) follow from the identities \((\text{vec}(A))^H \text{vec}(B) = \text{trace}(A^H B)\) and \(\text{tr}(ABC) = \text{tr}(CAB)\). Next, we derive the CRB for estimating $\theta$, which corresponds to the first diagonal element of \( {\mathbf{F}}^{-1} \), i.e.,
\begin{equation}
\text{CRB}(\theta) = [{\mathbf{F}}^{-1}]_{1,1} = \left[{\mathbf{F}}_{\theta\theta} - {\mathbf{F}}_{\theta\tilde{\alpha}}{\mathbf{F}}_{\tilde{\alpha}\tilde{\alpha}}^{-1} {\mathbf{F}}_{\tilde{\alpha}\theta} \right]^{-1},\label{crbphi}
\end{equation}
where ${\mathbf{F}}_{\theta\tilde{\alpha}}=({\mathbf{F}}_{\tilde{\alpha}\theta})^T$ since the FIM is a symmetric matrix. It is important to note that we do not use the Bayesian Cramér-Rao Bound (BCRB), which requires considering the distribution of the random parameter when deriving the elements of the FIM \footnote{At BCRB, we have: $\text{CRB}(\theta) = [{E[\mathbf{F}}_{\theta\theta}] + E[\frac{\partial \ln p_\theta(\theta)}{\partial \theta} \frac{\partial \ln p_\theta(\theta)}{\partial \theta^\text{T}} ]$-

$E[{\mathbf{F}}_{\theta\tilde{\alpha}}](E[{\mathbf{F}}_{\tilde{\alpha}\tilde{\alpha}}])^{-1} E[{\mathbf{F}}_{\tilde{\alpha}\theta}]]^{-1}$.}. Instead, we treat the CRB—derived for a fixed parameter—as a function of the random variables \(\mathbf{h}\), \(\theta\) and $\phi$. When the distributions of these variables are known, as in our case, we can evaluate \(E(\text{CRB})\) and \(P(\text{CRB})\) to analyze the behavior of the CRB. This approach of treating the CRB as a random variable is utilized in \cite{cramerraoboundonaerospaceandelectronicsystems, recentinsightsinthebayesiananddeterministic}.

To derive \( \dot{\mathbf{X}} = \frac{\partial \mathbf{X}}{\partial \theta} \) and, consequently, to obtain the element of the FIM in (\ref{elemntoffisher2}), there are two approaches:  1) Common Approximation: This approach, commonly used in the literature \cite{MIMOIntegratedSensingandCommunicationCRBRateTradeoff, CrameRaoBoundOptimizationforJoint, PhysicalLayerSecurityOptimizationWithCramér–RaoBoundMetric}, assumes a target-tracking stage of radar sensing where the target moves slowly enough for the BS to obtain rough estimates of the parameters \( c_3 \) and \( \theta \) from the previous stage \cite{CrameRaoBoundOptimizationforJoint, MIMOIntegratedSensingandCommunicationCRBRateTradeoff}. Based on these estimates, the BS designs \( \mathbf{X} \), assuming it to be fixed. Consequently, \( \dot{\mathbf{X}} = \frac{\partial \mathbf{X}}{\partial \theta} = 0 \) \cite[Appendix C]{TargetDetectionandLocalization};  2) Exact Derivation: In this approach, \( \mathbf{X} \) is not treated as a fixed parameter but as a function of random parameter \( \mathbf{h} \) and \( \theta \).  In the following, we derive \( P(\text{CRB}(\theta)>\epsilon) \) using both approaches. 
\begin{figure*}[t]
\normalsize
\begin{align}
{\mathbf{F}}_{\theta\theta} &= \frac{2|c_3|^2}{\sigma_R^2} \Re \left\{ ( \text{vec}(\mathbf{\dot{A}} \mathbf{X}))^H \text{vec}(\mathbf{\dot{A}} \mathbf{X})+( \text{vec}(\mathbf{\dot{A}} \mathbf{X}))^H \text{vec}(\mathbf{A} \mathbf{\dot{X}})+( \text{vec}(\mathbf{A} \mathbf{\dot{X}}))^H \text{vec}(\mathbf{\dot{A}} \mathbf{X})+( \text{vec}(\mathbf{A} \mathbf{\dot{X}}))^H \text{vec}(\mathbf{A} \mathbf{\dot{X}}) \right\}\nonumber\\
&\overset{a}{=}\frac{2 |c_3|^2}{\sigma_R^2}\left\{ L\text{tr}(\mathbf{\dot{A}} \mathbf{R}_x \mathbf{\dot{A}}^H)+\text{tr}(\mathbf{{A}} \mathbf{\dot{X}} \mathbf{{X}}^H\mathbf{\dot{A}}^H)+\text{tr}(\mathbf{{\dot{A}}} \mathbf{{X}} \mathbf{{\dot{X}}}^H\mathbf{{A}}^H)+\text{tr}(\mathbf{{A}} \mathbf{\dot{X}} \mathbf{{\dot{X}}}^H\mathbf{{A}}^H),\right\}\nonumber\\
{\mathbf{F}}_{\theta\bar{\alpha}} &= \frac{2}{\sigma_R^2} \Re \left\{ (c_3^* \text{vec}(\mathbf{\dot{A}} \mathbf{X})^H+c_3^* \text{vec}(\mathbf{{A}} \mathbf{\dot{X}})^H) [1, j] \otimes \text{vec}(\mathbf{A} \mathbf{X}) \right\}\overset{a}{=}\frac{2}{\sigma_R^2} \Re \left\{ Lc_3^* (\text{tr}(\mathbf{A} \mathbf{R}_x \mathbf{\dot{A}}^H)+\text{tr}(\mathbf{A} \mathbf{X}\mathbf{\dot{X}}^H \mathbf{A}^H) )[1, j] \right\},\nonumber\\
\tilde{\mathbf{F}}_{\bar{\alpha}\bar{\alpha}} &= \frac{2}{\sigma_R^2} \Re \left\{ ([1, j] \otimes \text{vec}(\mathbf{A} \mathbf{X}))^H ([1, j] \otimes \text{vec}(\mathbf{A} \mathbf{X})) \right\}\overset{a}{=}\frac{2}{\sigma_R^2} \Re \left\{ ([1, j]^H [1, j]) (\text{tr}(\mathbf{A} \mathbf{X})^H \mathbf{A} \mathbf{X}) \right\}\label{elemntoffisher2}.
\end{align}
\hrulefill
\end{figure*}

\textbf{Common approximation:} By setting \( \dot{\mathbf{X}} = \frac{\partial \mathbf{X}}{\partial \theta} = 0 \) at (\ref{elemntoffisher2}), we obtain: ${\mathbf{F}}_{\theta\theta} =\frac{2L |c_3|^2}{\sigma_R^2} \{\text{tr}(\mathbf{\dot{A}} \mathbf{R}_x \mathbf{\dot{A}}^H)\}$, ${\mathbf{F}}_{\theta\bar{\alpha}}=\frac{2L}{\sigma_R^2} \Re \left\{ c_3^* (\text{tr}(\mathbf{A} \mathbf{R}_x \mathbf{\dot{A}}^H))[1, j] \right\}$, and $\tilde{\mathbf{F}}_{\bar{\alpha}\bar{\alpha}}=\frac{2L}{\sigma_R^2} \text{tr}(\mathbf{A} \mathbf{R}_x \mathbf{A}^H) \mathbf{I}_2$. The CRB($\theta$) is derived in Appendix \ref{lemma4p}. Moreover, In Appendix \ref{lemma4p}, we propose a lower bound (proof in Appendix \ref{forsummeryproof}), an upper bound (proof in Appendix \ref{forsummeryproof2}), and an approximate value (proof in Appendix \ref{forsummeryproof3}) for \( P(\text{CRB}(\theta) > \epsilon) \). 

\textbf{Exact derivation:} 
We rewrite the received echo signal at the BS by substituting \(\mathbf{X}\) from (\ref{x}) into \(\mathbf{Y}_r\) and considering the orthogonality between \(\mathbf{a}\) and \(\mathbf{t}_i\) for \(i \in \{3, \dots, N\}\), as well as between \(\mathbf{a}\) and \(\mathbf{\tilde{h}}\). Additionally, using \(||\mathbf{a}||^2 = N\), we obtain:  $\mathbf{Y}_{r} = \alpha c_3 \mathbf{b}(\theta)\sqrt{N} \sqrt{P\tau} \mathbf{s}_u + \mathbf{Z}_{r}.$ In Appendix \ref{proofexactderivationlemma}, $\text{CRB}(\theta)$ and \( P(\text{CRB}(\theta) > \epsilon) \) are derived.  
\subsection{CCDF of $\text{CRB}(\phi)$ at the sensing eavesdropper}\label{opoftargeteav}
In the received echo signal at the sensing eavesdropper, \( \mathbf{Y}_{sr} \), we assume that \( c_4 \) is also an unknown but deterministic parameter. However, our primary parameter of interest is \( \phi \), which is to be estimated. The CCDF of CRB for the target’s angle with respect to the sensing eavesdropper is derived in Appendix \ref{proofcrbsensingeav}.
\subsection{Target Ergodic CRB at the BS and Sensing Eavesdropper}\label{targetergodiccrbatthebs}
Two approaches for deriving the CCDF of the CRB at the BS—namely, the common approximation and the exact derivation—lead to two different ergodic CRBs at the BS. In Appendix \ref{proofofderivingcommon}, we derived \( E[\text{CRB}(\theta)] \) using these two approaches. Moreover, \( E[\text{CRB}(\phi)] \) at the sensing eavesdropper is derived.
\section{Secrecy Ergodic Rate}\label{secrecyergodicrate}
In this section, we derive the ergodic secrecy rate, defined as $C_s = \left( \mathbb{E}_{\mathbf{h},\mathbf{a}(\theta)} \left[ \log(1 + \text{SINR}_u) \right] - \mathbb{E}_{\mathbf{h},\mathbf{a}(\theta),\mathbf{h}_e} \left[ \log(1 + \text{SINR}_e) \right] \right)^+.$ In Appendix \ref{proofpdfuser}, we derive the ergodic rate of the user, i.e., \( \mathbb{E}_{\mathbf{h},\mathbf{a}} \left[ \log(1 + \text{SINR}_u) \right] \). Furthermore, in Appendix \ref{proofpdfeav}, we derive the information leakage to the eavesdropper, namely \( \mathbb{E}_{\mathbf{h}_e, \mathbf{h}, \mathbf{a}(\theta)} \left[ \log(1 + \text{SINR}_e) \right] \).
\section{Numerical Results}\label{simulations}
Unless stated otherwise, we use the following parameters: \( N = 15 \), \( M = 17 \), \( N_e = 15 \), \( P = 10 \), \( \sigma_u = \sigma_r = 1 \), \( L = 30 \), \( c_1 = c_2 = \sqrt{0.001} \), \( \delta = 0.1 \), \( \alpha = 0.2 \), and \( c_3 = c_4 = 0.001 \). Simulations are averaged over 10,000 channel realizations. Fig. \ref{newfigures}(a) illustrates the ergodic rates at the user and the communication eavesdropper as functions of \( \tau \), the fraction of BS power allocated to user data. The ergodic secrecy rate remains positive and generally increases with \( \tau \), thanks to partial power allocation to artificial noise (AN). It peaks before slightly declining near \( \tau = 1 \), but stays positive. The user rate upper bound 2 is tight, and the results confirm Remark 2, showing that both the user and eavesdropper rates grow with \( \tau \).

Fig. \ref{newfigures}(b) presents the behavior of \( E[\text{CRB}(\theta)] \) and \( E[\text{CRB}(\phi)] \) with respect to \( \tau \). It is observed that the BS consistently achieves better angle estimation, as \( E[\text{CRB}(\phi)] \) is higher than both the exact and approximate values of \( E[\text{CRB}(\theta)] \) \footnote{As the number of antennas at the sensing eavesdropper increases, \( E[\text{CRB}(\phi)] \) decreases, which is evident from its expression derived in Section \ref{targetergodiccrbatthebs}.}. With increasing \( \tau \), exact CRBs decrease, while the approximate \( E[\text{CRB}(\theta)] \) increases—reflecting the assumption that both the data beam and AN contribute to sensing. The exact method, which accounts for AN being orthogonal to \( \mathbf{a}(\theta) \), shows that only the data beam affects estimation. As \( \tau \to 1 \), exact and approximate CRBs converge.

Fig. \ref{newfigures}(c) shows the complementary cumulative distribution functions (CCDFs) of the CRB, i.e., \( P(\text{CRB}(\phi) > \epsilon) \) and \( P(\text{CRB}(\theta) > \epsilon) \), plotted against \( 10 \log_{10}(\epsilon/10) \) for \( \tau = 0.76 \). As expected, the CCDFs decrease with increasing \( \epsilon \). The CRB at the sensing eavesdropper is consistently higher than both the exact and approximate CRBs at the BS. Furthermore, in the common approach, the upper bound, lower bound, and approximation are all closely aligned, demonstrating the tightness of these bounds. Monte Carlo results closely match the numerical derivations across all figures.
\begin{figure*}
    \centering
    \subfigure[Ergodic rate versus $\tau$]{\includegraphics[scale=.42]{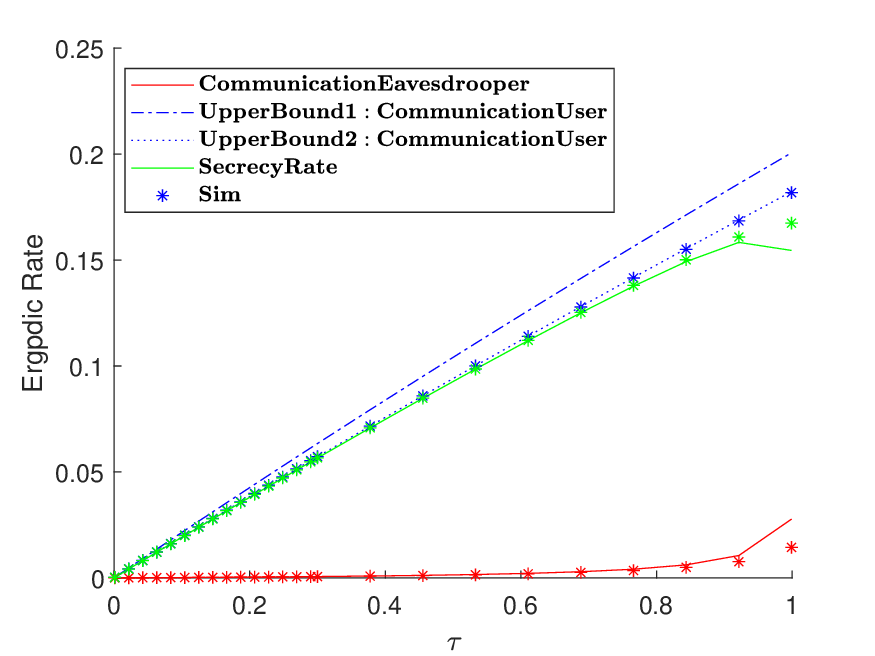}}\hspace{-0.5cm}\label{1}
    \subfigure[Ergodic CRB versus $\tau$]{\includegraphics[scale=.42]{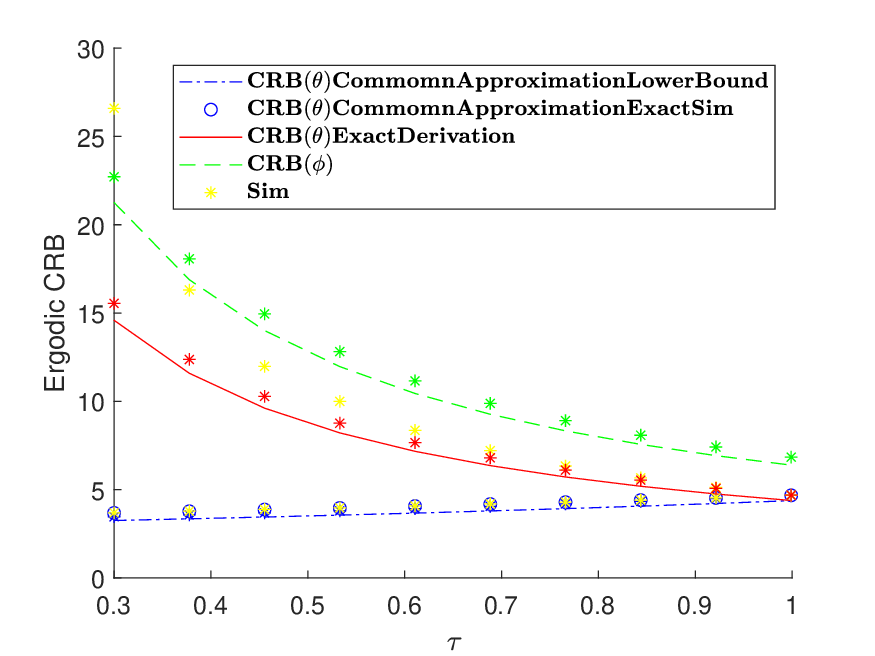}}\hspace{-0.5cm}\label{2}
    \subfigure[CCDF of CRB versus $10\times \log(\epsilon/10)$]{\includegraphics[scale=.42]{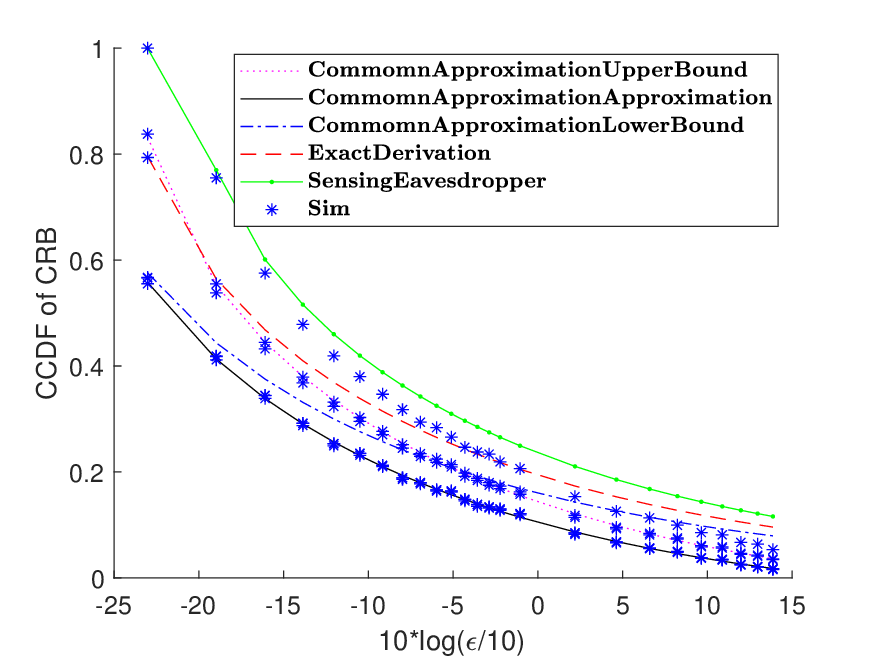}}\hspace{-0.5cm}\label{3}
    \caption{Ergodic rate, Ergodic CRB, and CCDF of CRB}
    \label{newfigures}
\end{figure*}
\bibliographystyle{IEEEtran}
\bibliography{ref}
\appendices
\section{Proof of Lemma \ref{lemma4}}\label{lemma4p}
 After some mathematical operation and using the FIM elements, (\ref{crbphi}) simplifies to (\ref{crb}). Then, by substituting (\ref{Rx}) and $\mathbf{A}(\theta)$ at (\ref{crb}), the CRB($\theta$) is derived in the following lemma:
\begin{lemma}\label{lemma4}
$\text{CRB}(\theta)$ is equal to:
\begin{align}
\frac{Q}{\gamma_1||\mathbf{b}'||^2N|\alpha|^2+\gamma_2M(||\mathbf{a}'||^2-\frac{(\sum_{i=1}^{N}-f'_i{t}_i)^2+(\sum_{i=1}^{N}f'_i{r}_i)^2}{K-\frac{1}{N}(T^2+R^2)}))}\label{crbsimplified}
\end{align}
where $\mathbf{b}'$ is the derivation of $\mathbf{b}(\theta)$ with respect to $\theta$. ${r}_i\triangleq \mathcal{R}(e^{jf_i}{h}_i)=|{h}_i|\cos(f_i+{{\phi}}_i)$, ${t}_i\triangleq \mathcal{I}(e^{jf_i}{h}_i)=|{h}_i|\sin(f_i+{{\phi}}_i)$, ${{k}}_i\triangleq |{h}_i|^2$, ${R}\triangleq \sum_{i=1}^{N}{r}_i$, ${T}\triangleq \sum_{i=1}^{N}{t}_i$, ${{K}}\triangleq \sum_{i=1}^{N}{{k}}_i$, $||\mathbf{b}'||^2=\frac{\pi^2\cos^2(\theta)M(M^2-1)}{12}$, $||\mathbf{a}'||^2=\frac{\pi^2\cos^2(\theta)N(N^2-1)}{12}$, $Q\triangleq\frac{\sigma^2_R}{2 \mid c_3 \mid ^2 L}$, $P\tau \triangleq \gamma_1$ and $\frac{(1-\tau)P}{N-2} \triangleq \gamma_2$.
\end{lemma}
\begin{figure*}[t]
\normalsize
\begin{align}
\text{CRB}(\theta)&=\frac{\sigma^2_R \text{Tr}(\mathbf{A}^H(\theta) \mathbf{A}(\theta) \mathbf{R}_x)}{2 \mid c_3 \mid ^2 L (\text{Tr}(\mathbf{A}^H(\theta) \mathbf{A}(\theta) \mathbf{R}_x) \text{Tr}(\dot{\mathbf{A}}^{H}(\theta) \dot{\mathbf{A}}^(\theta) \mathbf{R}_x)-\mid \text{Tr}(\dot{\mathbf{A}}^{H}(\theta) \mathbf{A}(\theta) \mathbf{R}_x)\mid^2)}. \label{crb}
\end{align}
\hrulefill
\end{figure*}
To prove this lemma, we begin by simplifying each term in the numerator and denominator of (\ref{crb}). By defining $P\tau \triangleq \gamma_1$ and $\frac{(1-\tau)P}{N-2} \triangleq \gamma_2$, we obtain:
\begin{align}
&\text{Tr}(\mathbf{A}^H(\theta) \mathbf{A}(\theta) \mathbf{R}_x)\overset{(a)}=\text{Tr}(\mathbf{a}\mathbf{b}^H\mathbf{b}\mathbf{a}^H ( 
\gamma_1\mathbf{t}_1\mathbf{t}^H_1+\gamma_2\sum_{i=3}^{N}\mathbf{t}_i\mathbf{t}^H_i)\nonumber\\
&\overset{(b)}=\gamma_1||\mathbf{b}||^2|\mathbf{a}^H\mathbf{t}_1|^2,\label{denom1}
\end{align}
where (a) follows from the use of $\mathbf{R}_X$ in (\ref{Rx}), and (b) follows from the property $\text{Tr}(abc) = \text{Tr}(bca) = \text{Tr}(cab)$ and the orthogonality between $\mathbf{a}$ and $\mathbf{t}_i$ for $i \in \{3, \dots, N\}$. Moreover, we have:
\begin{align}
&\text{Tr}(\mathbf{A}'^H(\theta) \mathbf{A}'(\theta) \mathbf{R}_x)=\nonumber\\
&\text{Tr}((\mathbf{a}\mathbf{b}'^H+\mathbf{a}'\mathbf{b}^H)(\mathbf{b}'\mathbf{a}^H+\mathbf{b}\mathbf{a}'^H) ( 
\gamma_1\mathbf{t}_1\mathbf{t}^H_1+\gamma_2\sum_{i=3}^{N}\mathbf{t}_i\mathbf{t}^H_i))\nonumber\\
&\overset{(a)}=\gamma_1(||\mathbf{b}'||^2|\mathbf{a}^H\mathbf{t}_1|^2+||\mathbf{b}||^2|\mathbf{a}'^H\mathbf{t}_1|^2)+\gamma_2||\mathbf{b}||^2\sum_{i=3}^{N}|\mathbf{a}'^H\mathbf{t}_i|^2,\label{denom2}
\end{align}
where (a) is due to the following:
\begin{align}
&\!\!\!\!\!\text{Tr}((\mathbf{b}'\mathbf{a}^H+\mathbf{b}\mathbf{a}'^H)\gamma_1\mathbf{t}_1\mathbf{t}_1^H(\mathbf{a}\mathbf{b}'^H+\mathbf{a}'\mathbf{b}^H))\nonumber\\
&\!\!\!\!\!\overset{(a)}=\gamma_1(\text{Tr}(\mathbf{a}^H\mathbf{t}_1\mathbf{t}^H_1\mathbf{a}\mathbf{b}'^H\mathbf{b}')+\text{Tr}(\mathbf{a}'^H\mathbf{t}_1\mathbf{t}^H_1\mathbf{a}'\mathbf{b}^H\mathbf{b})).\\
&\!\!\!\!\!\text{Tr}((\mathbf{b}'\mathbf{a}^H+\mathbf{b}\mathbf{a}'^H)\gamma_2(\sum_{i=3}^{N}\mathbf{t}_i\mathbf{t}_i^H)(\mathbf{a}\mathbf{b}'^H+\mathbf{a}'\mathbf{b}^H))\nonumber\\
&\!\!\!\!\!\overset{(b)}=\!\text{Tr}(\mathbf{a}'^H\gamma_2(\sum_{i=3}^{N}\mathbf{t}_i\mathbf{t}_i^H)\mathbf{a}'\mathbf{b}^H\mathbf{b})).
\end{align}
where (a) and (b) follow from the properties $\text{Tr}(abc) = \text{Tr}(bca) = \text{Tr}(cab)$ and $\text{Tr}(\mathbf{b}'\mathbf{a}^H\mathbf{t}_1\mathbf{t}_1^H\mathbf{a'}\mathbf{b}^H) = \text{Tr}(\mathbf{b}\mathbf{a}'^H\mathbf{t}_1\mathbf{t}_1^H\mathbf{a}\mathbf{b}'^H) = \text{Tr}(\mathbf{b}'\mathbf{a}^H\mathbf{t}_i\mathbf{t}_i^H\mathbf{a'}\mathbf{b}^H) = \text{Tr}(\mathbf{b}\mathbf{a}'^H\mathbf{t}_i\mathbf{t}_i^H\mathbf{a}\mathbf{b}'^H) = 0$ since $\mathbf{a}^H\mathbf{a}' = \mathbf{a}'^H\mathbf{a} = \mathbf{b}^H\mathbf{b}' = \mathbf{b}'^H\mathbf{b} = 0$
and $\mathbf{a}^H\mathbf{t}_i = 0 \quad \text{for} \quad i \in \{3, \dots, N\}.
$ Moreover:
\begin{align}
&\text{Tr}(\mathbf{A}'^H(\theta) \mathbf{A}(\theta) \mathbf{R}_x)=\nonumber\\
&\text{Tr}(\mathbf{b}\mathbf{a}^H( 
\gamma_1\mathbf{t}_1\mathbf{t}^H_1+\gamma_2\sum_{i=3}^{N}\mathbf{t}_i\mathbf{t}^H_i)(\mathbf{a}\mathbf{b}'^H+\mathbf{a}'\mathbf{b}^H))\nonumber\\
&\overset{(a)}=\gamma_1||\mathbf{b}||^2\mathbf{a}^H\mathbf{t}_1\mathbf{t}_1^H\mathbf{a}'\label{denom3}
\end{align}
where (a) follows from the properties $\text{Tr}(\mathbf{b}\mathbf{a}^H\mathbf{t}_1\mathbf{t}_1^H(\mathbf{a}\mathbf{b}'^H)) = \text{Tr}(\mathbf{b}\mathbf{a}^H\mathbf{t}_i\mathbf{t}_i^H(\mathbf{a}\mathbf{b}'^H)) = 0,$ $\text{Tr}(abc) = \text{Tr}(bca) = \text{Tr}(cab)$, and the fact that \(\mathbf{a}^H\mathbf{t}_i = 0\) for \(i \in \{3, \dots, N\}\). Additionally, since \(\mathbf{a}^H\mathbf{t}_1\mathbf{t}_1^H\mathbf{a}'\) is a scalar, the expression simplifies accordingly. Therefore, by utilizing (\ref{denom1}), (\ref{denom2}), (\ref{denom3}), and $\gamma^2||\mathbf{b}||^4\mathbf{a}^H\mathbf{t}_1\mathbf{t}_1^H\mathbf{a}'\mathbf{a}'^H\mathbf{t}_1\mathbf{t}_1^H\mathbf{a} = \gamma^2||\mathbf{b}||^4|\mathbf{a}^H\mathbf{t}_1|^2|\mathbf{a}'^H\mathbf{t}_1|^2,$ and defining $Q \triangleq \frac{\sigma^2_R}{2 |c_3|^2 L},$ we obtain:
\begin{align}
\text{CRB}(\theta)=\frac{Q}{\gamma_1||\mathbf{b}'||^2|\mathbf{a}^H\mathbf{t}_1|^2+\gamma_2||\mathbf{b}||^2\sum_{i=3}^{N}|\mathbf{a}'^H\mathbf{t}_i|^2}\label{crbnew1}
\end{align}
Moreover, by replacing \(\mathbf{t}_1 = \alpha \tilde{\mathbf{a}} + \beta \tilde{\mathbf{h}}\), we get \( |\mathbf{a}^H\mathbf{t}_1|^2 = |\alpha|^2 N \), where we use the orthogonality between \(\mathbf{a}\) and \(\tilde{\mathbf{h}}\), and \(|\mathbf{a}|^2 = N\). Furthermore, we have \(\sum_{i=3}^{N}|\mathbf{a}'^H\mathbf{t}_i|^2 = \mathbf{a}'^H (\mathbf{I} - \tilde{\mathbf{a}}\tilde{\mathbf{a}}^H - \tilde{\mathbf{h}}\tilde{\mathbf{h}}^H) \mathbf{a'}\), since the matrix \(\begin{bmatrix} \mathbf{\tilde{a}} & \mathbf{\tilde{h}} & \mathbf{G} \end{bmatrix}\), where \(\mathbf{G} = \begin{bmatrix} \mathbf{t}_3 & \dots & \mathbf{t}_N \end{bmatrix}\), forms an orthonormal basis for the \(N\)-dimensional space and is unitary. After some mathematical operations and replacing the definitions of \(\mathbf{\tilde{a}}\) and \(\mathbf{\tilde{h}}\), and noting that \(\mathbf{a}'(\theta)^H \mathbf{a}(\theta) = 0\), we get \(\sum_{i=3}^{N}|\mathbf{a}'^H\mathbf{t}_i|^2 = ||\mathbf{a}'||^2 - \frac{|\mathbf{a}'^H\mathbf{h}|^2}{||\mathbf{h}||^2 - \frac{1}{N} |\mathbf{a}^H \mathbf{h}|^2}\). Thus, by noting that \(||\mathbf{b}||^2 = M\), equation (\ref{crbnew1}) simplifies as:
\begin{align}
\text{CRB}(\theta)=\frac{Q}{\gamma_1||\mathbf{b}'||^2N|\alpha|^2+\gamma_2M(||\mathbf{a}'||^2-\frac{|\mathbf{a}'^H\mathbf{h}|^2}{||\mathbf{h}||^2-\frac{1}{N}|\mathbf{a}^H\mathbf{h}|^2}))}\label{crbnew2}
\end{align}
Moreover, by rewriting \(\mathbf{a}\) and \(\mathbf{h}\) element-wise and after performing some operations, the \(\text{CRB}(\theta)\) is derived as (\ref{crbsimplified}).
\begin{lemma}\label{forsummery}
A lower bound for $P(\text{CRB}(\theta)>\epsilon)$ is given by:  
\begin{align}  
&\frac{2}{\pi}\sin^{-1}(\sqrt{6}\sigma_R(\epsilon MN \pi ^2LP \mid c_3 \mid ^2\big(|\alpha|^2\tau(M^2-1) \label{crbb1zero}\\   
&+\frac{(N^2-1)(1-\tau)}{N-2}\big))^{-1/2})\nonumber,
\end{align}  
when $(\sqrt{6}\sigma_R(\epsilon MN \pi ^2LP \mid c_3 \mid ^2\big(|\alpha|^2\tau(M^2-1)+\frac{(N^2-1)(1-\tau)}{N-2}\big))^{-1/2})<1,$ and \( 1 \) otherwise. 
\end{lemma}
\begin{lemma}\label{forsummery2}
An upper bound for $P(\text{CRB}(\theta)>\epsilon)$ is given by:  
\begin{align}  
P_{Uc}(\epsilon)&=\frac{1}{\pi}\!\!\int_{0}^{\pi}\!\!\iiint_{\mathcal{D}(\theta, R,T, K)}\!\!\!\!\!\!f(R,T,K)\,dR\,dT\,dK\,d\theta, \label{crbb1one}  
\end{align}  
where $\mathcal{D}(\theta, R, T, K)= \frac{\frac{6\sigma_R^2}{L|c_3|^2\pi^2\cos^2(\theta)MN}}{\gamma_1|\alpha|^2(M^2-1)+\gamma_2(N^2-1)(1-\frac{1}{1-\frac{1}{N}\frac{(T^2+R^2)}{K}})}>\epsilon$ and \( f(R,T,K) \) is the PDF of a trivariate normal distribution with a mean vector \( N\mathbf{\mu_d} \) and covariance matrix \( N\mathbf{\Sigma_d} \) derived in Lemma \ref{lemma1i}, which is the product of three Gaussian distributions.
\end{lemma}
In general, integrating a multivariate normal PDF over an arbitrary interval lacks a closed-form solution. However, we have used ray-tracing method and the MATLAB toolbox presented in \cite{Amethodtointegrate}, which enables the integration of Gaussian distributions in any dimension, over any domain.
\begin{lemma}\label{forsummery3}
An approximation for $P(\text{CRB}(\theta)>\epsilon)$ is derived by replacing $\mathcal{\tilde{D}}(\theta, R, T, K) = \frac{\frac{6\sigma_R^2}{L|c_3|^2\pi^2\cos^2(\theta)MN}}{\gamma_1|\alpha|^2(M^2-1)+\gamma_2(N^2-1)(1-\frac{1}{K-\frac{1}{N}(T^2+R^2)})}>\epsilon$ instead of $\mathcal{{D}}(\theta, R, T, K)$ into (\ref{crbb1one}).
\end{lemma}
\section{Proof of Lemma \ref{forsummery}}\label{forsummeryproof}
If we ignore the term \(\big((\sum_{i=1}^{N}-f'_i\hat{t}_i)^2+(\sum_{i=1}^{N}f'_i\hat{r}_i)^2\big)\) in the denominator of CRB\((\theta)\) in (\ref{crbsimplified}), the total CRB will be smaller. The proof follows from the Cauchy–Schwarz inequality for the \(N\)-dimensional complex space, which states that:  
\begin{align}  
|\sum_{i=1}^{N}\!e^{jf_i}{h}_i|^2=R^2+T^2&<(\sum_{i=1}^{N}\!|e^{jf_i}|^2)(\sum_{i=1}^{N}\!|{h}_i|^2)\nonumber\\&=N(\sum_{i=1}^{N}\!|h_i|^2)=NK,  
\end{align}  
thus, the term \(\frac{(\sum_{i=1}^{N}-f'_i{t}_i)^2+(\sum_{i=1}^{N}f'_i{r}_i)^2}{K-\frac{1}{N}(T^2+R^2)}\) is positive. By ignoring the term \(\big((\sum_{i=1}^{N}-f'_i\hat{t}_i)^2+(\sum_{i=1}^{N}f'_i\hat{r}_i)^2\big)\), the denominator of CRB\((\theta)\) becomes larger.  
Thus, we obtain a lower bound on the CRB, referred to as LCRB, given by  
\[
\frac{6\sigma^2_R}{\cos^2(\theta)MN \pi ^2LP \mid c_3 \mid ^2\big(|\alpha|^2\tau(M^2-1)+\frac{(N^2-1)(1-\tau)}{N-2}\big)}.
\]  
Consequently,  
\begin{align}  
\!\!\!\!P_c(\epsilon)=P(\text{CRB}(\theta)>\epsilon)>P(\text{LCRB}(\theta)\!\!>\epsilon)\triangleq P_{Lc}(\epsilon).  
\end{align}  
Therefore, \( P_{Lc} \) provides a lower bound on the CCDF of the target angle CRB. Since \( \theta \) is uniformly distributed in the interval \([0,\pi]\), after calculating the CDF of \(\cos^2(\theta)\), (\ref{crbb1zero}) is derived.
\section{Proof of Lemma \ref{forsummery2}}\label{forsummeryproof2}
Using the Cauchy–Schwarz inequality for an \(N\)-dimensional complex space, we have:  
\begin{align}  
&(\sum_{i=1}^{N}-f'_i{t}_i)^2+(\sum_{i=1}^{N}f'_i{r}_i)^2=|\sum_{i=1}^{N}\!jf'_ie^{jf_i}{h}_i|^2\nonumber\\  
&<(\sum_{i=1}^{N}\!|jf'_i|^2)(\sum_{i=1}^{N}\!|e^{jf_i}{h}_i|^2)=(\sum_{i=1}^{N}\!|f'_i|^2)({{K}}),  
\end{align}  
where \(\sum_{i=1}^{N}\!|f'_i|^2=\frac{\pi^2\cos^2(\theta)N(N^2-1)}{12}\). By replacing \((\sum_{i=1}^{N}-f'_i{t}_i)^2+(\sum_{i=1}^{N}f'_i{r}_i)^2\) with \({{K}}(\sum_{i=1}^{N}\!|f'_i|^2)\) in the denominator of CRB\((\theta)\) in (\ref{crbsimplified}), we obtain an upper bound on the CRB, referred to as UCRB. Thus,  
\begin{align}  
P_c(\epsilon)\!=\!\!P(\text{CRB}(\theta)\!>\epsilon)\!<\!P(\text{UCRB}(\theta)\!>\!\epsilon)\!\triangleq P_{Uc}(\epsilon),\label{ostad}  
\end{align}  
where \( P_{Uc} \) provides an upper bound on the CCDF of the target angle CRB. Using (\ref{crbsimplified}) and (\ref{ostad}), along with \( ||\mathbf{a}'||^2=(\sum_{i=1}^{N}\!|f'_i|^2) \), and conditioning on \( \theta \) while noting that \( \theta \) and \( \mathbf{h} \) are independent, we obtain:  
\begin{align}  
&P_{Uc}(\epsilon)\!\!=\!\!\!\int_{0}^{\pi}\!\!\!\!\!\!P\left(\frac{\frac{6\sigma_R^2}{L|c_3|^2\pi^2\cos^2(\theta)MN}}{\gamma_1|\alpha|^2(M^2-1)+\gamma_2(N^2-1)(1-\frac{1}{1-\frac{1}{N}\frac{(T^2+R^2)}{K}})}\right.\nonumber\\  
&>\epsilon \Bigg| \theta \Bigg) f_{\theta}(\theta)d\theta. \label{up}  
\end{align}  
We note that in (\ref{up}), given \( \theta \), the only random variables are \( R \), \( T \), and \( K \). Thus, to compute the inner probability, we need to derive the joint probability density function (PDF) of \( R \), \( T \), and \( K \). These variables (as well as \( r_i \), \( t_i \), and \( k_i \)) are not independent since they are functions of \( |h_i| \) (which follows a Rayleigh distribution), \( \phi_i \) (which is uniformly distributed over \( [0, 2\pi) \)), and \( f_i \) (a function of \( \theta \)), as provided in Lemma \ref{lemma4}.  
By conditioning on \( \theta \), \( f_i \), for all \( i=1,...,N \), is treated as constant in the following derivations. We define \( N \) random vectors, \( \mathbf{d}_i=[r_i, t_i, k_i]^T \in \mathbb{R}^{3 \times 1} \), for \( i=1,...,N \). For any pair of \( j \) and \( i\neq j \), the vectors \( \mathbf{d}_j \) and \( \mathbf{d}_i \) are i.i.d. because \( h_i \)'s and \( \phi_i \)'s are i.i.d.  
Thus, by applying the multidimensional central limit theorem (CLT)\cite{enwiki:1155685628}, when \( N \) is large \footnote{Section \ref{simulations} shows that for \( N>9 \), the multidimensional CLT holds. Additionally, in \cite{Aunifiedperformanceframeworkfor}, the authors demonstrate via simulations that CLT accuracy for a one-dimensional random variable holds for \( N>8 \).}, we have: $\sqrt{N}[\frac{1}{N}(\sum_{i=1}^{N}\mathbf{d}_i)-\mathbf{\mu_d}]\overset{d}{\rightarrow} \mathcal{N}_3(\mathbf{0},\mathbf{\Sigma_d}),$ which implies $\sum_{i=1}^{N}\mathbf{d}_i\overset{d} \rightarrow \mathcal{N}_3(N\mathbf{\mu_d}, N\mathbf{\Sigma_d}),$ where \( \mathbf{\mu_d} \) and \( \mathbf{\Sigma_d} \) are the mean vector and covariance matrix of \( \mathbf{d}_i \), respectively (identical for all \( i=1,...,N \)). Therefore, $[R, T, K]^T\overset{(d)}{\rightarrow} \mathcal{N}_3(N\mathbf{\mu_d},N\mathbf{\Sigma_d}).$\footnote{Monte Carlo simulations confirm the numerical approximation.}  
By determining \( \mathbf{\mu_d} \) and \( \mathbf{\Sigma_d} \) using Lemma \ref{lemma1i}, we derive the joint PDF of \( R \), \( T \), and \( K \).  
\begin{lemma} \label{lemma1i}
\(\mathbf{\mu_d}\) and \(\mathbf{\Sigma_d}\) are given by: $\mathbf{\mu}_d=[0, 0, 1]^T$ and  
\begin{align}  
\mathbf{\Sigma}_d&=\begin{bmatrix}  
\frac{1}{2}& 0 & 0 \\  
0& \frac{1}{2} & 0 \\  
0& 0 & 1   
\end{bmatrix}.\label{rtk}  
\end{align}
\end{lemma}
\begin{proof}
First, we derive ${{\mathbf{\mu_d}}} = [E[{{r}}_i], E[{{t}}_i], E[{{k}}_i]]$. Since \( |h_i|^2 \) follows an exponential distribution with a mean of 1, we have $\mathbb{E}[|h_i|^2] = 1$. The terms in $r_i$ and $k_i$ contain cosine functions involving the random variables \( \phi_i \). The expectations of these cosine terms are zero because the angles are uniformly distributed. Thus, we obtain ${{\mathbf{\mu_d}}} = [0, 0, 1]$.  
Next, to calculate each element of the covariance matrix $\mathbf{{{\Sigma_d}}}$, we use the independence of $|h_i|$ and $\theta$, along with $E(|h_i|^4) = 2$. Additionally, we utilize the following identities: $\cos(A)\cos(B) = \frac{1}{2} [\cos(A - B) + \cos(A + B)], \cos(A)\sin(A) = \frac{1}{2} \sin(2A), \sin^2(A) = \frac{1}{2} (1 - \cos(2A)).$
Due to space limitations, we omit the detailed derivations.
\end{proof}
From Lemma \ref{lemma1i}, we see that the joint conditional PDF of \( R \), \( T \), and \( K \) is independent of \( \theta \) and \( \phi \). In fact, \( R \), \( T \), and \( K \) are independent Gaussian random variables since the off-diagonal elements of the covariance matrix are zero. Using this property, along with the assumption that \( \theta \) is uniformly distributed, and defining the domain: $\mathcal{D}(\theta, R, T, K)= \frac{\frac{6\sigma_R^2}{L|c_3|^2\pi^2\cos^2(\theta)MN}}{\gamma_1|\alpha|^2(M^2-1)+\gamma_2(N^2-1)(1-\frac{1}{1-\frac{1}{N}\frac{(T^2+R^2)}{K}})}>\epsilon.$, the proof is complete.
\section{Proof of Lemma \ref{forsummery3}}\label{forsummeryproof3}
When \( N \) is large, the law of large numbers implies that:  
\begin{align}
(\sum_{i=1}^{N}-f'_i{t}_i)^2+(\sum_{i=1}^{N}f'_i{r}_i)^2\overset{(p)}{\rightarrow} E\{(\sum_{i=1}^{N}-f'_i{t}_i)^2+(\sum_{i=1}^{N}f'_i{r}_i)^2 \}.
\end{align}  
Thus, by replacing \( (\sum_{i=1}^{N}-f'_i{t}_i)^2+(\sum_{i=1}^{N}f'_i{r}_i)^2 \) in the denominator of CRB(\(\theta\)) in (\ref{crbsimplified}) with  
\begin{align}  
E\{(\sum_{i=1}^{N}-f'_i{t}_i)^2+(\sum_{i=1}^{N}f'_i{r}_i)^2 \} \overset{(a)}{=} \sum_{i=1}^{N}(f'_i)^2E\{{{k_i}}\}=\sum_{i=1}^{N}(f'_i)^2,  
\end{align}  
where (a) follows from the independence of \( t_i \) from \( t_j \) and \( r_i \) from \( r_j \), we obtain an approximation of CRB(\(\theta\)), denoted as ACRB.  
Thus, we have:  
\begin{align}
P_c(\epsilon)\!=\!P(\text{CRB}(\theta)\!>\!\epsilon)\!\approx\! P(\text{ACRB}(\theta)\!>\!\epsilon)\!\triangleq\! P_{Ac}(\epsilon),\label{ap}
\end{align}  
where \( P_{Ac} \) serves as an approximation of the CCDF of the
target angle CRB. Following the same approach as in Appendix \ref{forsummeryproof2}, $P_{Ac}(\epsilon)$, is derived by replacing $\mathcal{\tilde{D}}(\theta, R, T, K) = \frac{\frac{6\sigma_R^2}{L|c_3|^2\pi^2\cos^2(\theta)MN}}{\gamma_1|\alpha|^2(M^2-1)+\gamma_2(N^2-1)(1-\frac{1}{K-\frac{1}{N}(T^2+R^2)})}>\epsilon.$ instead of $\mathcal{{D}}(\theta, R, T, K)$ into (\ref{crbb1one}).
\section{Proof of Lemma \ref{exactderivationlemma}}\label{proofexactderivationlemma}
\begin{lemma}\label{exactderivationlemma}
$\text{CRB}(\theta) =6\sigma_R^2(L|c_3|^2 P\tau N|\alpha|^2\pi^2\cos^2(\theta)M(M^2-1))^{-1}$. Thus, $P(\text{CRB}(\theta)>\epsilon)=\frac{2}{\pi}\sin^{-1}(\sqrt{6}\sigma_R(\epsilon L|c_3|^2 P\tau N|\alpha|^2\pi^2M(M^2-1))^{-1/2})$ when $(\sqrt{6}\sigma_R(\epsilon L|c_3|^2 P\tau N|\alpha|^2\pi^2M(M^2-1))^{-1/2})<1,$ and \( 1 \) otherwise.
\end{lemma}
To begin the proof, we express the vectorized form of $\mathbf{Y}_{r}$ as follows: $\tilde{\tilde{\mathbf{y}}}_{r} = \text{vec}(\mathbf{Y}_{r}) = \tilde{\tilde{\mathbf{u}}} + \tilde{\mathbf{n}},$ where \( \tilde{\tilde{\mathbf{u}}} = \alpha c_3 \sqrt{N} \sqrt{P\tau} \text{vec}(\mathbf{b}(\theta)\mathbf{s}_u ) \) and \( \tilde{\mathbf{Z}_{r}} = \text{vec}(\mathbf{Z}_{r}) \sim \mathcal{CN}(0, \sigma^2_r\mathbf{I}_{ML}) \). Thus, we have: $\tilde{\tilde{\mathbf{y}}}_{r}\sim \mathcal{CN}(\tilde{\tilde{\mathbf{u}}}, \sigma^2_r\mathbf{I}_{ML})$. Furthermore, $\frac{\partial \tilde{\tilde{\mathbf{u}}}}{\partial \theta} = \alpha c_3 \sqrt{N} \sqrt{P\tau} \text{vec}(\dot{\mathbf{b}}(\theta)\mathbf{s}_u )$ and $\frac{\partial \tilde{\tilde{\mathbf{u}}}}{\partial \bar{\alpha}} = \alpha \sqrt{N} \sqrt{P\tau}[1, j] \otimes \text{vec}({\mathbf{b}}(\theta)\mathbf{s}_u )$. Following the same approach as (\ref{elemntoffisher2}) and noting that \(\frac{1}{L}E\{\mathbf{s}_u\mathbf{s}_u^H\} \approx 1\), we obtain: ${\mathbf{F}}_{\theta\theta} =\frac{2L P\tau N|\alpha|^2|c_3|^2}{\sigma_R^2} \{\text{tr}(\mathbf{\dot{b}} \mathbf{\dot{b}}^H)\}$, ${\mathbf{F}}_{\theta\bar{\alpha}}=\frac{2LP\tau N|\alpha|^2}{\sigma_R^2} \Re \left\{ c_3^* (\text{tr}(\mathbf{b} \mathbf{\dot{b}}^H))[1, j] \right\}$, and $\tilde{\mathbf{F}}_{\bar{\alpha}\bar{\alpha}}=\frac{2LP\tau N|\alpha|^2}{\sigma_R^2} \text{tr}(\mathbf{b} \mathbf{b}^H) \mathbf{I}_2$. Based on (\ref{crbphi}) and after some mathematical manipulation, we have: \(\text{CRB}(\theta) = \sigma_R^2(2L|c_3|^2 P\tau N|\alpha|^2\left( \text{tr}(\mathbf{\dot{b}}\mathbf{\dot{b}}^H) - \frac{|\text{tr}(\mathbf{b}\dot{\mathbf{b}}^H)|^2}{\text{tr}(\mathbf{b}\mathbf{b}^H)} \right))^{-1}\).
Then, using $\text{Tr}(ab)=\text{Tr}(ba)$, \(\mathbf{b'}^H(\theta)\mathbf{b}(\theta) = 0\), and $||\mathbf{b}'||^2=\pi^2\cos^2(\theta)M(M^2-1)/12$, the proof is complete. By comparing this result with Lemma \ref{forsummery}, we observe that the outage probability (OP) of \( \text{CRB}(\theta) \) in the exact derivation is greater than the lower bound of the OP of \( \text{CRB}(\theta) \) in the common approximation scenario.
\section{Proof of Lemma \ref{crbsensingeav}}\label{proofcrbsensingeav}
\begin{lemma}\label{crbsensingeav}
$\text{CRB}(\phi) = \frac{\sigma^2_R}{2 |c_4|^2 L \gamma_1 ||\mathbf{c}'||^2 |\alpha|^2 N}$ where $||\mathbf{c}'||^2 = \frac{\pi^2 \cos^2(\phi) N_e (N_e^2 - 1)}{12}$ and $P(\text{CRB}(\phi)>{\epsilon})=\frac{2}{\pi}\sin^{-1}(\sqrt{6}\sigma_R({\epsilon} N_eN \pi ^2LP \mid c_4 \mid ^2|\alpha|^2\tau(N_e^2-1))^{-1/2})$ when $(\sqrt{6}\sigma_R({\epsilon} N_eN \pi ^2LP \mid c_4 \mid ^2|\alpha|^2\tau(N_e^2-1))^{-1/2})<1$, and $1$ otherwise.
\end{lemma} 
\textbf{Proof:} Based on the received echo signal at the sensing eavesdropper which is $\mathbf{Y}_{sr} = c_4 \mathbf{c}(\phi)\mathbf{a}(\theta)^H \mathbf{X} + \mathbf{Z}_{sr},$ where $\mathbf{X}=\sqrt{P\tau}\mathbf{t}_1 \mathbf{s}_u+ \sqrt{P(1-\tau)}\mathbf{G} \mathbf{V} = \sqrt{P\tau}\mathbf{t}_1 \mathbf{s}_u+ \sqrt{P(1-\tau)}\sum_{i=3}^{N}\mathbf{t}_i\mathbf{v}_i\in \mathbb{C}^{N\times L}$, we obtain the Fisher information matrix (FIM) for estimating \( \xi=[\phi,\mathcal{R}(c_4),\mathcal{I}(c_4)]^T \in \mathbb{R}^{3 \times 1} \) to facilitate the derivation of the $\text{CRB}(\phi)$. Let \( \mathbf{B}(\theta,\phi) =\mathbf{c}(\phi) \mathbf{a}(\theta)^T \), the received echo signal at the sensing eavesdropper can be rewritten as $\mathbf{Y}_{sr} = c_4 \mathbf{B}(\theta,\phi) \mathbf{X} + \mathbf{Z}_{sr}.$ For notational convenience, in the sequel we drop \( \theta \) and $\phi$ in \( \mathbf{B}(\theta,\phi) \), \( \mathbf{a}(\theta) \), and \( \mathbf{c}(\phi) \). By vectorizing $\mathbf{Y}_{sr}$, we have: $\tilde{\mathbf{y}}_{sr} = \text{vec}(\mathbf{Y}_{sr}) = \tilde{\mathbf{u}}_e + \tilde{\mathbf{n}},$ where \( \tilde{\mathbf{u}}_e = c_4 \text{vec}(\mathbf{B}\mathbf{X}) \) and \( \tilde{\mathbf{Z}_{sr}} = \text{vec}(\mathbf{Z}_{sr}) \sim \mathcal{CN}(0, \sigma^2_r\mathbf{I}_{N_eL}) \). By assuming that the eavesdropper is strong and it know $\mathbf{X}$ in the following we derive FIM. we have: $\tilde{\mathbf{y}}_{sr}\sim \mathcal{CN}(\tilde{\mathbf{u}}_e, \sigma^2_r\mathbf{I}_{N_eL})$. By defining $\bar{\alpha}_e=[\mathcal{R}(c_4),\mathcal{I}(c_4)]^T \in \mathbb{R}^{2 \times 1}$, we have $\frac{\partial \tilde{\mathbf{u}}_e}{\partial \phi} = c_4 \text{vec}(\mathbf{\dot{B}} \mathbf{X})+c_4 \text{vec}(\mathbf{B} \mathbf{\dot{X}})$ and $\frac{\partial \tilde{\mathbf{u}}_e}{\partial \bar{\alpha}_e} = [1, j] \otimes \text{vec}(\mathbf{B} \mathbf{X})$. Accordingly, the element of ${\mathbf{F}}$ is obtained as (\ref{elemntoffisher2}) by substituting $\mathbf{A}$ with $\mathbf{B}$, $\bar{\alpha}$ with $\bar{\alpha}_e$, $\theta$ with $\phi$, and $c_3$ with $c_4$, respectively. Also, \( \dot{\mathbf{B}} = \frac{\partial \mathbf{B}}{\partial \phi} \) and \( \dot{\mathbf{X}} = \frac{\partial \mathbf{X}}{\partial \phi} \) denoting the partial derivative of \( \mathbf{B} \) and $\mathbf{X}$ w.r.t. \( \phi \), respectively. Since \(\mathbf{X}\) is a function of \(\mathbf{a}(\theta)\) and \(\mathbf{h}\) (based on the construction of the basis for the \(N\)-dimensional space in Section \ref{systemmodel}), and its derivative with respect to \(\phi\) is zero. Thus, (\ref{elemntoffisher2}) is simplified to:
${\mathbf{F}}_{\phi\phi}=\frac{2L |c_4|^2}{\sigma_R^2} \{\text{tr}(\mathbf{\dot{B}} \mathbf{R}_x \mathbf{\dot{B}}^H)\}$, ${\mathbf{F}}_{\phi\bar{\alpha}_e} =\frac{2L}{\sigma_R^2} \Re \left\{ c_4^* (\text{tr}(\mathbf{B} \mathbf{R}_x \mathbf{\dot{B}}^H))[1, j] \right\}$, and ${\mathbf{F}}_{\bar{\alpha}_e\bar{\alpha}_e}=\frac{2L}{\sigma_R^2} \text{tr}(\mathbf{B} \mathbf{R}_x \mathbf{B}^H) \mathbf{I}_2$. Thus, similar to (\ref{crbphi}), we have: $\text{CRB}(\phi) = \left[{\mathbf{F}}_{\phi\phi} - {\mathbf{F}}_{\phi\bar{\alpha}_e}{\mathbf{F}}_{\bar{\alpha}_e\bar{\alpha}_e}^{-1} {\mathbf{F}}_{\bar{\alpha}_e\phi} \right]^{-1}$, where ${\mathbf{F}}_{\phi\bar{\alpha}_e}=({\mathbf{F}}_{\bar{\alpha}_e\phi})^T$. Thus, after some mathematical manipulation we have: $\text{CRB}(\theta) = \frac{\sigma_R^2}{2L|c_4|^2 \left( \text{tr}(\mathbf{\dot{B}} \mathbf{R}_x \mathbf{\dot{B}}^H) - \frac{|\text{tr}(\mathbf{B} \mathbf{R}_x \dot{\mathbf{B}}^H)|^2}{\text{tr}(\mathbf{B} \mathbf{R}_x \mathbf{B}^H)} \right)}$, where $\mathbf{\dot{B}}=\mathbf{\dot{c}}(\phi)\mathbf{a}^H(\theta)$ and $\mathbf{R}_x=\frac{1}{L}\mathbf {X}\mathbf{X}^H\approx 
P\tau\mathbf{t}_1\mathbf{t}^H_1+\frac{(1-\tau)P}{N-2}\sum_{i=3}^{N}\mathbf{t}_i\mathbf{t}^H_i$. By defining $P\tau \triangleq \gamma_1$ and $\frac{(1-\tau)P}{N-2} \triangleq \gamma_2$, We have:
\begin{align}
&\text{tr}(\mathbf{B} \mathbf{R}_x \mathbf{B}^H)=\text{Tr}(\mathbf{c}\mathbf{a}^H ( 
\gamma_1\mathbf{t}_1\mathbf{t}^H_1+\gamma_2\sum_{i=3}^{N}\mathbf{t}_i\mathbf{t}^H_i)\mathbf{a}\mathbf{c}^H)\nonumber\\
&\overset{(a)}=\gamma_1||\mathbf{c}||^2|\mathbf{a}^H\mathbf{t}_1|^2,\nonumber\\
&\text{tr}(\mathbf{\dot{B}} \mathbf{R}_x \mathbf{\dot{B}}^H)=\text{Tr}(\mathbf{\dot{c}}(\phi)\mathbf{a}^H(\theta) ( 
\gamma_1\mathbf{t}_1\mathbf{t}^H_1+\gamma_2\sum_{i=3}^{N}\mathbf{t}_i\mathbf{t}^H_i)\mathbf{a}(\theta)\mathbf{\dot{c}}^H(\phi))\nonumber\\
&\overset{(a)}=\gamma_1(||\mathbf{c}'||^2|\mathbf{a}^H\mathbf{t}_1|^2),\nonumber\\
&\text{tr}(\mathbf{B} \mathbf{R}_x \dot{\mathbf{B}}^H)=\text{Tr}(\mathbf{c}\mathbf{a}^H( 
\gamma_1\mathbf{t}_1\mathbf{t}^H_1+\gamma_2\sum_{i=3}^{N}\mathbf{t}_i\mathbf{t}^H_i)\mathbf{a}(\theta)\mathbf{\dot{c}}^H(\phi))\nonumber\\
&\overset{a}{=}0\label{hello}
\end{align}
where (a)s are due to $\text{Tr}(abc)=\text{Tr}(bca)=\text{Tr}(cab)$, the orthogonality between $\mathbf{a}$ and $\mathbf{t}_i$ for $i \in 3,...,N$, and \(\mathbf{c'}^H(\phi)\mathbf{c}(\phi) = 0\). Moreover, by replacing $\mathbf{t}_1= \alpha \tilde{\mathbf{a}} + \beta \tilde{\mathbf{h}}$, we have: $|\mathbf{a}^H\mathbf{t}_1|^2=|\alpha|^2N$ where we have used the orthogonality between $\mathbf{a}$ and $\tilde{\mathbf{h}}$ and $|\mathbf{a}|^2=N$. Therefore, we have: $(\text{CRB}(\phi) = \frac{\sigma^2_R}{2 |c_4|^2 L \gamma_1 ||\mathbf{c}'||^2 |\alpha|^2 N}$ where $||\mathbf{c}'||^2 = \frac{\pi^2 \cos^2(\phi) N_e (N_e^2 - 1)}{12}$. Since \(\phi\sim \mathcal{U}(-\pi/2,\pi/2)\), after calculating the CDF of \(\cos^2(\phi)\) \footnote{Any distribution of \(\phi\) can be used; one just needs to derive the CDF of \(\cos^2(\phi)\) for that distribution.}, the proof is complete.

\textbf{Remark1:} We note that setting \( \tau = 0 \) in the exact derivation of \( \text{CRB}(\theta) \) or Lemma \ref{crbsensingeav} yields \( \text{CRB}(\theta) = \text{CRB}(\phi)= \infty \). However, in Lemmas \ref{lemma4}, \ref{forsummery}, \ref{forsummery2}, and \ref{forsummery3}, a finite \( \text{CRB}(\theta) \) is obtained. This implies that when the entire BS power is allocated to AN, conventional CRB derivations assume AN contributes to sensing. In contrast, the exact method shows that AN does not aid target sensing at either the BS or the sensing eavesdropper, as it is orthogonal to \( \mathbf{a}(\theta) \) and therefore vanishes in the received echo at both locations.
\section{Ergodic CRB at the BS and Sensing eavesdropper}\label{proofofderivingcommon}
\textbf{Common approximation of $E[\text{CRB}(\theta)]$:} The approximated ergodic CRB at the BS is:
\begin{align}
\!\!\!\!E[\text{CRB}(\theta)]&\overset{(a)}=\int_{0}^{\infty}\!\!\!\!P(\text{CRB}(\theta)>t)dt\overset{(b)}\sim \int_{0}^{\infty} \!\!\!\!P_{Ac}(t)dt,\label{ergodiccrb}
\end{align}
where (a) and (b) follow from statistical properties, and \( P_{Ac}(\epsilon) = P(\text{CRB}(\theta) > \epsilon) \) is derived in Lemma \ref{forsummery3}\footnote{By substituting \( P_{Uc}(t) \), derived in (\ref{crbb1one}), one can obtain an upper bound for the ergodic CRB.}. A lower bound for the ergodic CRB at the BS is obtained by substituting \( P_{Lc}(t) \), derived in (\ref{crbb1zero}), into (\ref{ergodiccrb}). Thus, by truncating the domain to \( \theta \in [-\frac{\pi}{2} + \delta, \frac{\pi}{2} - \delta] \) for a small \( \delta > 0 \), a closed-form lower bound for the ergodic CRB at the BS is: $
\mathbb{E}[\text{CRB}(\theta)] = 
\frac{12\sigma^2_R}{ MN \pi ^3LP |c_3|^2\left(|\alpha|^2\tau(M^2-1)+\frac{(N^2-1)(1-\tau)}{N-2}\right)} \tan\left(\frac{\pi}{2} - \delta\right).
$

\textbf{Exact derivation of $E[\text{CRB}(\theta)]$:} The exact ergodic CRB at the BS is derived using the same approach as described above, along with the result from Lemma \ref{exactderivationlemma}, yielding:
$ \mathbb{E}[\text{CRB}(\theta)] = 
(\frac{12\sigma_R^2}{L|c_3|^2 P\tau N|\alpha|^2\pi^3M(M^2-1)}) \tan\left(\frac{\pi}{2} - \delta\right).$

\textbf{\( E[\text{CRB}(\phi)] \) at the sensing eavesdropper:} By substituting \( \epsilon \) with \( t \) in Lemma \ref{crbsensingeav} and using the result in (\ref{ergodiccrb}), and by truncating the domain \( \phi \in [-\frac{\pi}{2} + \delta, \frac{\pi}{2} - \delta] \) for a small \( \delta > 0 \), the closed-form expression for the ergodic CRB at the sensing eavesdropper is: $
\mathbb{E}[\text{CRB}(\phi)] = 
\frac{12\sigma^2_R}{ N_eN \pi ^3LP |c_4|^2|\alpha|^2\tau(N_e^2-1)} \tan\left(\frac{\pi}{2} - \delta\right)
$.
\section{Proof of Lemma \ref{lemmapdfuser}}\label{proofpdfuser}
\begin{lemma}\label{lemmapdfuser}
Exact ergodic rate at the user is:
\begin{align}
E_{\mathbf{h,a}} \!\!\left[ \log(1 + \text{SINR}_u) \right]&\!\!=\!\!\!\int_{0}^{\infty}\!\!\!\!\!\!\iiint_{\mathcal{\tilde{D}}(R,T, K, W)}\!\!\!\!\!\!\!\!\!\!\!\!\!\!\!\!\!\!\!\!\!\!\!\!\!\!f(R,T,K, W)\,dR\,dT\,dK\,dWdt\label{final}
\end{align}
where $\mathcal{\tilde{D}}(R, T, K, W)= \frac{P\tau|c_1|^2}{\sigma^2_u}[\frac{|\alpha|^2 - |\beta|^2}{N} (R^2 + T^2) + |\beta|^2 K + 2 \frac{|\alpha \beta|}{\sqrt{N}} W \sqrt{K - \frac{1}{N} (R^2 + T^2)}]>2^t-1$, and $f(R,T, K, W)$ is the PDF of a multivariate normal distribution with a mean vector of $N\mathbf{\tilde{\mu_d}}$ and a covariance matrix of $N\mathbf{\tilde{\Sigma_d}}$ given at (\ref{covuser}). Moreover, two upper bound for $E_{\mathbf{h,a}} \left[ \log(1 + \text{SINR}_u) \right]$ are: $\log\left(1 + \frac{P \tau |c_1|^2}{\sigma_u^2} (|\alpha|^2 + |\beta|^2 (N - 1))\right)$ and $\frac{e^{\sigma_u^2/P \tau |c_1|^2}}{(N-1)!} \cdot G_{2,3}^{3,1}\left(a \middle| \begin{array}{c} 
0, 0 \\ 
0, -1, N-1 
\end{array} \right),$
where \( G \) is the Meijer G-function.
\end{lemma}
We note that the (\ref{final})) can be evaluated using the ray-tracing method and the MATLAB toolbox presented in \cite{Amethodtointegrate}.

\textbf{Proof:} First, we derive the exact expression for the ergodic rate at the user, followed by the derivation of its upper bounds. Based on (\ref{x}), and given the signal received at the user as described in Section \ref{systemmodel}, along with the orthogonality of \(\mathbf{v}_i\) and \(\mathbf{s}_u\), the SINR at the user is given by: 
$\text{SINR}_u=\frac{P\tau|c_1|^2|\mathbf{h}^H\mathbf{t}_1|^2}{\sigma^2+\frac{P(1-\tau)\sum_{i=3}^{N}|\mathbf{h}^H\mathbf{t}_i|^2}{N-2}}$. Due to the orthogonality between \(\mathbf{h}\) and \(\mathbf{t}_i\), this expression simplifies to:  
$\text{SINR}_u=\frac{P\tau|c_1|^2}{\sigma^2_u}|\mathbf{h}^H\mathbf{t}_1|^2.
$ First, we rewrite \( \text{SINR}_u = \frac{P \tau |c_1|^2}{\sigma^2_u} |\mathbf{h}^H \mathbf{t}_1|^2 \), where \( \mathbf{t}_1 = \alpha \tilde{\mathbf{a}} + \beta \tilde{\mathbf{h}} \), \( \tilde{\mathbf{a}} = \frac{\mathbf{a}}{||\mathbf{a}||} \), and \( \tilde{\mathbf{h}} = \frac{\mathbf{h} - (\tilde{\mathbf{a}}^H \mathbf{h}) \tilde{\mathbf{a}}}{||\mathbf{h} - (\tilde{\mathbf{a}}^H \mathbf{h}) \tilde{\mathbf{a}}||} \). After some mathematical operations, we obtain the following:
\begin{align}
\text{SINR}_u&=\frac{P\tau|c_1|^2}{\sigma^2_u}[(\frac{|\alpha|^2-|\beta|^2}{N})|\mathbf{a}^H\mathbf{h}|^2+|\beta|^2|\mathbf{h}|^2\nonumber\\&
+\frac{2}{\sqrt{N}}\mathcal{R}(\alpha^*\beta(\mathbf{a}^H\mathbf{h}))\sqrt{|\mathbf{h}|^2-\frac{|\mathbf{a}^H\mathbf{h}|^2}{N}}]\nonumber\\&
\overset{(a)}{=} \frac{P\tau|c_1|^2}{\sigma^2_u}[\frac{|\alpha|^2 - |\beta|^2}{N} (R^2 + T^2) + |\beta|^2 K \nonumber\\
&+ 2 \frac{|\alpha \beta|}{\sqrt{N}} W \sqrt{K - \frac{1}{N} (R^2 + T^2)}]\label{simplifieduser}
\end{align}
where (a) is due to rewriting \( \mathbf{a} \) and \( \mathbf{h} \) element-wise, using the definition of \( R \), \( T \), and \( K \) as in Lemma \ref{lemma4}, and \( W \triangleq \sum_{i=1}^{N} w_i \) where \( w_i \triangleq |h_i| \cos(f_i + \phi_i + \phi_{\beta} - \phi_{\alpha}) \), with \( \phi_{\alpha} \) and \( \phi_{\beta} \) being the phases of \( \alpha \) and \( \beta \), respectively. We note that in (\ref{simplifieduser}), the random variables (RVs) are \( R \), \( T \), \( K \), and \( W \), which are functions of \( |h_i| \) (where \( |h_i| \) follows a Rayleigh distribution), \( \phi_i \) (where \( \phi_i \) is the phase of \( h_i \) and is uniformly distributed over \( [0, 2\pi) \)), and \( f_i \) (a function of \( \theta \), where \( \theta \) is uniformly distributed over \( [0, \pi) \)). Moreover, we have:
\begin{align}
E_{\mathbf{h,a}} \left[ \log(1 + \text{SINR}_u) \right]&\overset{(a)}{=}E_{\mathbf{a}}\left[E_{\mathbf{h}} \left[ \log(1 + \text{SINR}_u) |\theta\right]\right]\nonumber\\
&\overset{(b)}{=}E_{\mathbf{a}}\left[\int_{0}^{\infty}\!\!\!\!P(\text{SINR}_u>2^t-1|\theta)dt\right]\label{simplifieduser2}
\end{align}
where (a) is due to the independence of \( \mathbf{h} \) and \( \mathbf{a} \), noting that the elements of \( \mathbf{a} \) are a function of \( \theta \), and (b) is due to statistical properties. To calculate the inner probability in (\ref{simplifieduser2}), we need to derive the joint PDF of \( R \), \( T \), \( K \), and \( W \). We note that \( R \), \( T \), \( K \), and \( W \) (as well as \( r_i \), \( t_i \), \( K_i \), and \( w_i \)) are not independent, as they are functions of \( |h_i| \), \( \phi_i \), and \( \theta \). Conditioned on \( \theta \), \( f_i \), for all \( i = 1, \dots, N \), will be treated as constant. We define \( N \) random vectors \( \mathbf{\tilde{d}}_i = [r_i, t_i, k_i, w_i]^T \in \mathbb{R}^{4 \times 1} \) for \( i = 1, \dots, N \). Then, following the same approach as in Appendix \ref{forsummeryproof2}, \( [R, T, K, W]^T \overset{(d)}{\rightarrow} \mathcal{N}_4(N \mathbf{\tilde{\mu_d}}, N \mathbf{\tilde{\Sigma_d}}) \), where \( \mathbf{\tilde{\mu_d}} \) and \( \mathbf{\tilde{\Sigma_d}} \) are the mean vector and covariance matrix of \( \mathbf{\tilde{d}}_i \) (the same for all \( i = 1, \dots, N \)). We note that the covariance matrix and mean vector of the joint RVs \( R \), \( T \), and \( K \) are derived with the aid of Lemma \ref{lemma1i}. Thus, to derive \( \mathbf{\tilde{\mu_d}} \) and \( \mathbf{\tilde{\Sigma_d}} \), we just need to derive the mean of \( w_i \) and the joint covariance between \( w_i \) and \( t_i \), \( r_i \), and \( k_i \). Using trigonometric formulas and following the same approach as in the proof of Lemma \ref{lemma1i}, we have:
\begin{align}
\mathbf{\tilde{\mu}}_d\!\!&=\!\![0, 0, 1, 0]^T \nonumber\\
\mathbf{\tilde{\Sigma_d}}\!\!&=\!\!\!\begin{bmatrix}
\frac{1}{2}& 0 & 0 & \frac{1}{2}\cos(\phi_{\alpha}-\phi_{\beta}) \\
0& \frac{1}{2} & 0 & \frac{1}{2}\sin(\phi_{\alpha}-\phi_{\beta})\\
0& 0 & 1 & 0 \\
\frac{1}{2}\cos(\phi_{\alpha}-\phi_{\alpha}) &  \frac{1}{2}\sin(\phi_{\alpha}-\phi_{\beta}) & 0 & \frac{1}{2}\!\!
\end{bmatrix}\label{covuser}
\end{align}
Thus, the joint PDF of \( R \), \( T \), \( K \), and \( W \), and consequently the ergodic rate of the user, is independent of \( \theta \). Therefore, (\ref{simplifieduser2}) simplifies to:
\begin{align}
E_{\mathbf{h,a}} \left[ \log(1 + \text{SINR}_u) \right]&\overset{(a)}{=}\int_{0}^{\infty}\!\!\!\!P(\text{SINR}_u>2^t-1)dt\label{simplifieduser3}
\end{align}
By defining the domain $\mathcal{\tilde{D}}(R, T, K, W)$ as in Lemma \ref{lemmapdfuser}, we have:
\begin{align}
P(\text{SINR}_u>2^t-1)&=\!\!\iiint_{\mathcal{\tilde{D}}(R,T, K, W)}\!\!\!\!\!\!\!\!\!\!\!\!\!\!\!\!\!\!\!\!f(R,T,K, W)\,dR\,dT\,dK\,dW\label{simplifieduser4}
\end{align}
where \( f(R,T,K,W) \) is the PDF of a multivariate normal distribution with mean vector \( N \mathbf{\tilde{\mu_d}} \) and covariance matrix \( N \mathbf{\tilde{\Sigma_d}} \) given in (\ref{covuser}). Thus, by using (\ref{simplifieduser4}) in (\ref{simplifieduser3}), the proof is complete. Next, we derive the upper and lower bounds for the user's ergodic rate.

\textbf{Upper Bound1:} $|\mathbf{h}^H \mathbf{t}_1|^2 \leq \|\mathbf{h}\|^2,$ so $\text{SINR}_u \leq \frac{P \tau |c_1|^2}{\sigma_u^2} \|\mathbf{h}\|^2.$ Since \(\|\mathbf{h}\|^2 \sim \text{Gamma}(N, 1)= \frac{x^{N-1} e^{-x}}{(N-1)!}, \quad x \geq 0\), an upper bound for $E_{\mathbf{h,a}} \left[ \log(1 + \text{SINR}_u) \right]$ is:
\begin{align}
&\mathbb{E}[\log(1 + \frac{P \tau |c_1|^2}{\sigma_u^2} \|\mathbf{h}\|^2)]\!\!=\!\!\!\!\int_0^\infty \!\!\!\! \!\!\!\!\log(1 + \!\!\frac{P \tau |c_1|^2}{\sigma_u^2} x) \cdot \frac{x^{N-1} e^{-x}}{(N-1)!} \, dx\nonumber\\
&=\frac{e^{\sigma_u^2/P \tau |c_1|^2}}{(N-1)!} \cdot G_{2,3}^{3,1}\left(a \middle| \begin{array}{c} 
0, 0 \\ 
0, -1, N-1 
\end{array} \right),
\end{align}
where \( G \) is the Meijer G-function. For arbitrary \( N \), using Taylor expansion around \( \mathbb{E}[|h|^2] = N \) we have: $\mathbb{E}[\log(1 + \frac{P \tau |c_1|^2}{\sigma_u^2}|h|^2)] \approx \log(1 + \frac{P \tau |c_1|^2}{\sigma_u^2} N) - \frac{a^2}{(1 + \frac{P \tau |c_1|^2}{\sigma_u^2} N)^2} \cdot \frac{N}{2}$.

\textbf{Upper Bound 2:} We have $|\mathbf{h}^H \mathbf{t}_1|^2 = |\alpha|^2 |\tilde{\mathbf{a}}^H \mathbf{h}|^2 + |\beta|^2 (\|\mathbf{h}\|^2 - |\tilde{\mathbf{a}}^H \mathbf{h}|^2) + 2 \text{Re}\{\alpha^* \beta (\tilde{\mathbf{a}}^H \mathbf{h}) \sqrt{\|\mathbf{h}\|^2 - |\tilde{\mathbf{a}}^H \mathbf{h}|^2}\}.$ Moreover, \(|\tilde{\mathbf{a}}^H \mathbf{h}|^2 \sim \text{Exp}(1)\), so \(\mathbb{E}[|\tilde{\mathbf{a}}^H \mathbf{h}|^2] = 1\).Since \(\|\mathbf{h}\|^2 \sim \text{Gamma}(N, 1)\), so \(\mathbb{E}[\|\mathbf{h}\|^2] = N\). Furthermore, the expectation of $\mathbb{E}[\text{Re}\{\alpha^* \beta (\tilde{\mathbf{a}}^H \mathbf{h}) \sqrt{\|\mathbf{h}\|^2 - |\tilde{\mathbf{a}}^H \mathbf{h}|^2}\}]$ is zero, due to the following geometric intuition. We express \(\mathbf{h}\) in a basis where \(\tilde{\mathbf{a}}\) is the first basis vector. Let: $\mathbf{h} = (\tilde{\mathbf{a}}^H \mathbf{h}) \tilde{\mathbf{a}} + \mathbf{h}_\perp,$ where \(\mathbf{h}_\perp\) is the component orthogonal to \(\tilde{\mathbf{a}}\). Then: $\|\mathbf{h}\|^2 = |\tilde{\mathbf{a}}^H \mathbf{h}|^2 + \|\mathbf{h}_\perp\|^2.$
Thus: $\mathbb{E}[\text{Re}\{\alpha^* \beta (\tilde{\mathbf{a}}^H \mathbf{h}) \sqrt{\|\mathbf{h}\|^2 - |\tilde{\mathbf{a}}^H \mathbf{h}|^2}\}] = \mathbb{E}[\text{Re}\{\alpha^* \beta (\tilde{\mathbf{a}}^H \mathbf{h})\}|\mathbf{h}_\perp\|] =0.$ The equality is due to independency of \(\mathbf{h}_\perp\) and \(\tilde{\mathbf{a}}^H \mathbf{h}\), and \(\tilde{\mathbf{a}}^H \mathbf{h}\) is symmetric in phase. Finally, we have: $\mathbb{E}[\text{SINR}_u] = \frac{P \tau |c_1|^2}{\sigma_u^2} (|\alpha|^2 + |\beta|^2 (N - 1)).$ thus, using Jensen's inequality, we have:
\begin{align}
\mathbb{E}[\log(1 + \text{SINR}_u)] \leq \log\left(1 + \mathbb{E}[\text{SINR}_u]\right)\nonumber\\
= \log\left(1 + \frac{P \tau |c_1|^2}{\sigma_u^2} (|\alpha|^2 + |\beta|^2 (N - 1))\right).
\end{align}
\textbf{Asymptotic Behavior for Large \(N\)}
For massive MIMO (\(N \gg 1\)):\(\|\mathbf{h}\|^2 \approx N\), so: Upper bound: \(\log(1 + \text{SINR}_u) \approx \log\left(1 + \frac{P \tau |c_1|^2}{\sigma_u^2} N\right)\). Lower bound: \(\log(1 + \text{SINR}_u) \approx \log\left(1 + \frac{P \tau |c_1|^2}{\sigma_u^2} |\beta|^2 N\right)\).

\textbf{Remark2:} We note from Lemma \ref{lemmapdfuser} that the ergodic rate at the user is an increasing function of \( \tau \). Additionally, by differentiating (\ref{final2}) with respect to \( \tau \) and observing that \( \frac{dI}{d\tau} > 0 \) for \( \tau \in [0,1] \), we conclude that the ergodic rate at the communication eavesdropper also increases with \( \tau \). Therefore, allocating more power to data benefits both the user and the eavesdropper. Whether a positive ergodic secrecy rate is achieved will be verified in the numerical results section.
\section{Proof of Lemma \ref{lemmapdfeav}}\label{proofpdfeav}
\begin{lemma}\label{lemmapdfeav}  
The exact ergodic rate at the communication eavesdropper is given by:  
\begin{align}
E_{\mathbf{h_e,h},\mathbf{a}(\theta)}\left[ \log(1 + \text{SINR}_e) \right]\!\!&\overset{(a)}{=}\!\!\!\int_{0}^{\infty}\!\!\!\!\! e^{(-\frac{T}{2C_1})}(1 + \frac{TC_2}{C_1})^{-(N-2)}dt,\label{final2}
\end{align}  
where \(T=2^t-1\), $C_1 = \frac{P \tau |c_2|^2}{\sigma^2}$, and $C_2 = \frac{P(1-\tau)/(N-2)}{\sigma^2}.$ 
\end{lemma}  
We note that (\ref{final2}) can be expressed in terms of the incomplete Gamma functions; however, these representations are not computationally simpler than (\ref{final2}) itself, as it can be efficiently evaluated numerically using a single integral function in MATLAB.

\textbf{Proof:} Based on (\ref{x}), and given the received signal at the communication eavesdropper as described in Section \ref{systemmodel}, along with the orthogonality of \(\mathbf{v}_i\) and \(\mathbf{s}_u\), the SINR at the communication eavesdropper is given by:  
\begin{align}
\text{SINR}_e\!\!&\overset{(a)}{=}\frac{P\tau|c_2|^2|\mathbf{h}^H_e\mathbf{t}_1|^2}{\sigma^2+\frac{P(1-\tau)\sum_{i=3}^{N}|\mathbf{h}^H_e\mathbf{t}_i|^2}{N-2}},\label{sinr2i}
\end{align} 
We note that the matrix \(\begin{bmatrix} \mathbf{\tilde{a}} & \mathbf{\tilde{h}} & \mathbf{G} \end{bmatrix}\), where \(\mathbf{G}=\begin{bmatrix} \mathbf{t}_3 & \dots & \mathbf{t}_N \end{bmatrix}\), is unitary, and the entries of \(\mathbf{h}_e\) are i.i.d. zero-mean circularly symmetric complex normal random variables. Since the generation of \(\begin{bmatrix} \mathbf{\tilde{a}} & \mathbf{\tilde{h}} & \mathbf{G} \end{bmatrix}\) is entirely determined by the realizations of \(\mathbf{h}\) and \(\mathbf{a}\), which are independent of \(\mathbf{h}_e\) (as stated in Section \ref{systemmodel}), we conclude that \(\mathbf{h}_e\) and \(\begin{bmatrix} \mathbf{\tilde{a}} & \mathbf{\tilde{h}} & \mathbf{G} \end{bmatrix}\) are mutually independent.  
Thus, the distribution of \(\mathbf{h}^H_e\begin{bmatrix} \mathbf{\tilde{a}} & \mathbf{\tilde{h}} & \mathbf{G} \end{bmatrix}\) follows the same distribution as \(\mathbf{h}_e\), i.e., $\mathbf{h}^H_e\begin{bmatrix} \mathbf{\tilde{a}} & \mathbf{\tilde{h}} & \mathbf{G} \end{bmatrix} \sim \mathcal{CN}(0,\mathbf{I}).$ As a result, we can express $\mathbf{h}^H_e\mathbf{t}_1 = \alpha\mathbf{h}^H_e\mathbf{\tilde{a}} + \beta \mathbf{h}^H_e\mathbf{\tilde{h}},$ where \(\mathbf{h}^H_e\mathbf{\tilde{a}} \sim \mathcal{CN}(0,1)\) and \(\mathbf{h}^H_e\mathbf{\tilde{h}} \sim \mathcal{CN}(0,1)\). Since this expression is a linear combination of independent complex Gaussian random variables, it follows a complex Gaussian distribution with mean $\alpha E[\mathbf{h}^H_e\mathbf{\tilde{a}}] + \beta E[\mathbf{h}^H_e\mathbf{\tilde{h}}] = 0$ and variance $E[|\alpha\mathbf{h}^H_e\mathbf{\tilde{a}} + \beta \mathbf{h}^H_e\mathbf{\tilde{h}}|^2] = |\alpha|^2 + |\beta|^2 = 1.$ Thus, \(\mathbf{h}^H_e\mathbf{t}_1 \sim \mathcal{CN}(0,1)\), which implies \(|\mathbf{h}^H_e\mathbf{t}_1|^2 \sim \chi^2_2\).  
Furthermore, for \(i \in \{3, \dots, N\}\), we have \(\mathbf{h}^H_e\mathbf{t}_i \sim \mathcal{CN}(0,1)\), which are i.i.d. random variables. Therefore, the sum \(\sum_{i=3}^{N} |\mathbf{h}^H_e\mathbf{t}_i|^2\) follows a Chi-Square distribution with \(2(N-2)\) degrees of freedom, i.e., $\sum_{i=3}^{N} |\mathbf{h}^H_e\mathbf{t}_i|^2 \sim \chi^2_{2(N-2)}.
$ Moreover, since $\mathbf{h}^H_e\begin{bmatrix} \mathbf{\tilde{a}} & \mathbf{\tilde{h}} & \mathbf{t}_3 & \dots & \mathbf{t}_N \end{bmatrix} \sim \mathcal{CN}(0,\mathbf{I}),$ the random variables \(|\mathbf{h}^H_e\mathbf{t}_1|^2\) and \(\sum_{i=3}^{N} |\mathbf{h}^H_e\mathbf{t}_i|^2\) are independent. By defining \( C_1 = \frac{P \tau |c_2|^2}{\sigma^2}, \quad C_2 = \frac{P(1-\tau)/(N-2)}{\sigma^2} \), \( X \triangleq |\mathbf{h}^H_e \mathbf{t}_1|^2 \sim \chi^2_2 \), \( Y \triangleq \sum_{i=3}^{N} |\mathbf{h}^H_e \mathbf{t}_i|^2 \sim \chi^2_{2(N-2)} \), where \( X \) and \( Y \) are independent, we have \( \text{SINR}_e = \frac{C_1 X}{1 + C_2 Y} \). Thus:
\begin{align}
E_{X,Y} \left[ \log(1 + \text{SINR}_e) \right]&\overset{(a)}{=}\int_{0}^{\infty}\!\!\!\!P\left(X > \frac{T + TC_2Y}{C_1}\right)dt\label{simplifiedeav1}
\end{align}
where \( T = 2^t - 1 \), and (a) is due to the statistical property. Since \( X \sim \text{Exp}(1/2) \), its complementary CDF (tail probability) is: \( P(X > x) = e^{-x/2}, \quad x \geq 0 \). Thus, we get: \( P\left( X > \frac{T + TC_2Y}{C_1} | Y \right) = \exp\left( - \frac{T + TC_2Y}{2C_1} \right) \). Therefore, we have:
\begin{align}
&P\left(X > \frac{T + TC_2Y}{C_1}\right)\overset{(a)}{=}\mathbb{E}_Y\left[\exp\left(-\frac{T + TC_2Y}{2C_1}\right)\right]\nonumber\\
&=\exp\left(-\frac{T}{2C_1}\right) \cdot \mathbb{E}_Y\left[\exp\left(-\frac{TC_2Y}{2C_1}\right)\right]\nonumber\\
&\overset{(b)}{=}\exp\left(-\frac{T}{2C_1}\right)\left(1 + \frac{TC_2}{C_1}\right)^{-(N-2)}\label{simplifiedeav2}
\end{align}
(a) follows from statistical properties, and (b) holds because the term \( \mathbb{E}_Y\left[ \exp\left( - \frac{TC_2Y}{2C_1} \right) \right] \) is the moment-generating function (MGF) of \( Y \), evaluated at \( s = - \frac{TC_2}{2C_1} \). For \( Y \sim \text{Gamma}(N-2, 2) \), the MGF is given by \( M_Y(s) = \mathbb{E}_Y\left[ e^{sY} \right] = \left( 1 - \theta s \right)^{-k} \), where \( \theta = 2 \) is the scale parameter and \( k = N-2 \) is the shape parameter. Thus, replacing (\ref{simplifiedeav2}) into (\ref{simplifiedeav1}), the proof is complete.
\end{document}